\documentclass[a4paper]{article}
\usepackage{amsfonts}
\usepackage{amsmath}
\usepackage[margin=2cm]{geometry}
\usepackage{graphicx}
\usepackage{mathtools}

\newtheorem{theorem}{Theorem}

\newtheorem{remark}[theorem]{Remark}

\newenvironment{proof}[1][Proof]{\noindent\textbf{#1.} }{\ \rule{0.5em}{0.5em}}
\input{tcilatex}

\begin{document}

\title{Bifurcation analysis of rotating axially compressed imperfect nano-rod%
}
\author{Teodor M. Atanackovi\'{c}\thanks{
Department of Mechanics, Faculty of Technical Sciences, University of Novi
Sad, Trg D. Obradovi\'{c}a 6, 21000 Novi Sad, Serbia, atanackovic@uns.ac.rs}%
, Ljubica Oparnica\thanks{
Faculty of Education in Sombor, University of Novi Sad, Podgori\v{c}ka 4,
25000 Sombor, Serbia, ljubica.oparnica@pef.uns.ac.rs}, Du\v{s}an Zorica%
\thanks{
Mathematical Institute, Serbian Academy of Arts and Sciences, Kneza Mihaila
36, 11000 Belgrade, Serbia, dusan\textunderscore zorica@mi.sanu.ac.rs and
Department of Physics, Faculty of Sciences, University of Novi Sad, Trg D.
Obradovi\'{c}a 4, 21000 Novi Sad, Serbia}}
\maketitle

\vspace{-0.55cm}
\begin{abstract}
\noindent Static stability problem for axially compressed rotating nano-rod
clamped at one and free at the other end is analyzed by the use of
bifurcation theory. It is obtained that the pitchfork bifurcation may be
either super- or sub-critical. Considering the imperfections in rod's shape
and loading, it is proved that they constitute the two-parameter universal
unfolding of the problem. 
Numerical analysis also revealed that for non-locality parameters having higher value than the critical one 
interaction curves have two branches, so that for a single critical value of angular velocity 
there exist two critical values of horizontal force.

\end{abstract}

\section{Introduction and problem formulation}

The problem of static stability of cantilevered rotating axially compressed
rod displaying non-local effects is studied through the bifurcation theory,
extending the results presented in \cite{AZ-1}, where the Euler method of
adjacent equilibrium configuration is used to obtain critical values of the
angular velocity and intensity of the horizontal axial force acting on the
tip of rod's free end. The obtained critical values are shown to represent
the bifurcation points by using the Crandall-Rabinowitz theorem. Further,
the Lyapunov-Schmidt reduction method is applied in order to obtain
bifurcation equation corresponding to the non-linear equilibrium equations
of rotating compressed rod and it is shown that the problem admits pitchfork
bifurcation. Imperfections in shape, represented by the existence of a small
initial deformation of the rod, and imperfections in loading, represented by
the existence of a force of small intensity acting perpendicularly to rod's
axis on the tip of the rod, are also taken into account and it is proved
that the selected imperfections constitute the universal unfolding of the
problem. Moreover, the results presented in \cite{AZ-1} are extended by
finding the degenerate odd buckling modes for high values of non-locality
parameter. Considering the non-locality effects, included through the stress
gradient Eringen moment-curvature constitutive relation, the results of \cite%
{Ata}, where the same problem is analyzed in the case of Bernoulli-Euler
constitutive equation, are extended as well.

The buckling problem of a rotating compressed rod, described by the elastic
moment-curvature constitutive equation, is considered in \cite%
{Ata3,Bodnar,Lakin,CYWang}, while in \cite{Ata1,Ata2} the rod is allowed to
have variable cross section and extensible axis, and in \cite{Varadi} there
are additional rigid bodies attached to the rod. Static stability problem of
a non-local rotating compressed rod, described by the Eringen stress
gradient constitutive model, is studied in \cite{ANVZ,AZ-1} for the
clamped-clamped and clamped-free rod, while in \cite{Aranda} a non-local
clamped-free rod rotating about the axis perpendicular to rod's axis is
considered. The application of non-local theory in the static and dynamic
stability problems of different types of rods is quite extensive, see the
review articles \cite{ArashWang,EltaherKhaterEmam,ThaiVoNguyenKim} and book 
\cite{Elishakoff}.

Consider a rectangular Cartesian coordinate system $xOy$ forming a plane $%
\Pi $ that rotates about the $x$-axis with the constant angular velocity $%
\omega .$ Placed in its undeformed state in plane $\Pi ,$ an inextensible
rod of length $L$ and initial curvature $R_{0},$ changing along the rod, is
fixed in the origin of a coordinate system at one of its ends, while its
other end is free. Being in the relative equilibrium in plane $\Pi ,$ the
rod rotates and under the influence of inertial force it may lose its
stability and attain the relative equilibrium in the bent configuration, as
shown in Figure \ref{fig-1}. 
\begin{figure}[tbph]
\begin{center}
\includegraphics[width=0.525\columnwidth]{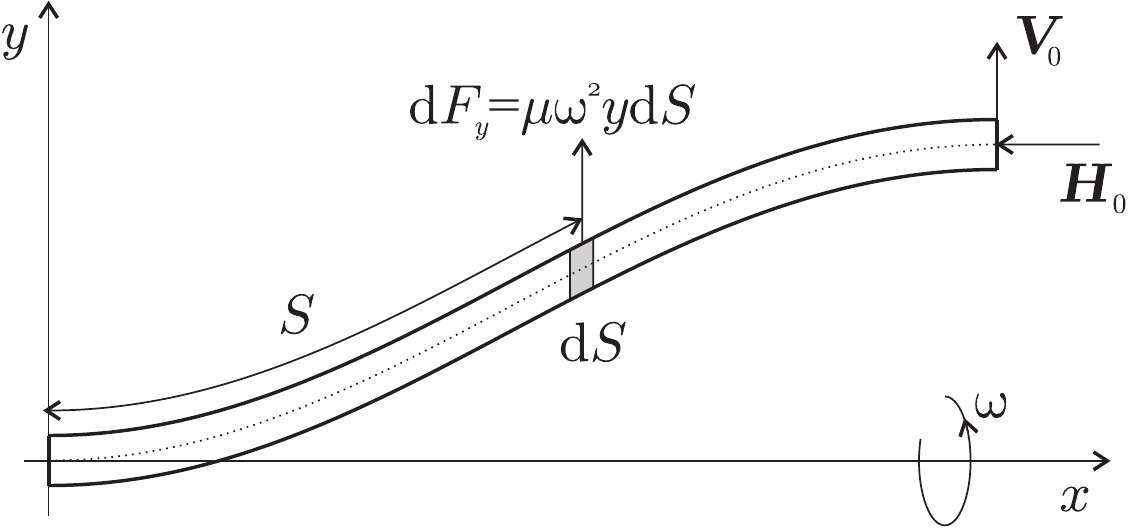}
\end{center}
\caption{Rotating axially compressed imperfect nano-rod.}
\label{fig-1}
\end{figure}

Differential equations and geometrical relations describing the relative
equilibrium in plane $\Pi $ are:%
\begin{gather}
H^{\prime }=0,\;\;\;\;V^{\prime }=-\mu \omega ^{2}y,\;\;\;\;M^{\prime
}=-V\cos \theta +H\sin \theta ,  \label{1a} \\
x^{\prime }=\cos \theta ,\;\;\;\;y^{\prime }=\sin \theta ,  \label{1b}
\end{gather}%
see \cite{a-ster}, where $H,$ $V,$ $M,$ $x,$ $y$ and $\theta $ are functions
of rod's arc length $S\in \left[ 0,L\right] $ and $\left( \cdot \right)
^{\prime }=\frac{\mathrm{d}}{\mathrm{d}S}\left( \cdot \right) ,$ with $H$
and $V$ being components of the contact force in an arbitrary cross-section
along $x$- and $y$-axis respectively, $M$ being the bending moment, $x$ and $%
y$ denoting the coordinates of an arbitrary point of a rod, and $\theta $
denoting the angle between the $x$-axis and tangent to rod's axis, while the
constant mass density per unit length of the rod is denoted by $\mu .$

The rod is assumed to display non-local effects and the moment-curvature
constitutive equation is assumed in the form of the Eringen stress-gradient
type model of non-locality as 
\begin{equation}
M-\ell ^{2}M^{\prime \prime }=EI\,\left( \frac{1}{R}-\frac{1}{R_{0}}\right)
,\;\;\text{with}\;\;\frac{1}{R}=\theta ^{\prime }=\frac{y^{\prime \prime }}{%
\sqrt{1-\left( y^{\prime }\right) ^{2}}},\text{ and }\frac{1}{R_{0}}=\frac{%
y_{0}^{\prime \prime }}{\sqrt{1-\left( y_{0}^{\prime }\right) ^{2}}},
\label{1c}
\end{equation}%
where $\frac{1}{R}$ and $\frac{1}{R_{0}}$ are the curvatures at equilibrium
and initial configuration as functions of arc-length $S$, while the
constants are: modulus of elasticity $E,$ moment of inertia of cross-section 
$I,$ and length-scale parameter $\ell .$ More on the Eringen type
stress-gradient constitutive equations can be found in \cite{eringen}.

System of equations (\ref{1a}) - (\ref{1c}) is subject to boundary conditions%
\begin{equation}
x\left( 0\right) =0,\;\;y\left( 0\right) =0,\;\;\theta (0)=0,\;\;H\left(
L\right) =-H_{0},\;\;V\left( L\right) =V_{0},\;\;M\left( L\right) =0,
\label{1d}
\end{equation}%
corresponding to the configuration shown in Figure \ref{fig-1}. Note that (%
\ref{1a})$_{1}$ and (\ref{1d})$_{4}$ imply $H\left( S\right) =-H_{0}.$

Dimensionless variables and parameters 
\begin{gather*}
t=\frac{S}{L},\;\;\bar{x}=\frac{x}{L},\;\;\bar{y}=\frac{y}{L},\;\;\bar{R}=%
\frac{R}{L},\;\;\bar{R}_{0}=\frac{R_{0}}{L},\;\;v=\frac{VL^{2}}{EI},\;\;m=%
\frac{ML}{EI},\;\;\kappa =\left( \frac{\ell }{L}\right) ^{2}, \\
\lambda _{1}=\frac{\mu \omega ^{2}L^{4}}{EI},\;\;\lambda _{2}=\frac{%
H_{0}L^{2}}{EI},\;\;\alpha _{1}=\frac{1}{\left\Vert R_{0}\right\Vert
_{L^{\infty }\left[ 0,1\right] }},\;\;\alpha _{2}=\frac{V_{0}L^{2}}{EI},
\end{gather*}%
where $\bar{R}_{0}\left( t\right) =\left\Vert \bar{R}_{0}\right\Vert
_{L^{\infty }\left[ 0,1\right] }\rho _{0}\left( t\right) ,$ with $\left\Vert 
\bar{R}_{0}\right\Vert _{L^{\infty }\left[ 0,1\right] }=\sup_{t\in \left[ 0,1%
\right] }\left\vert \bar{R}_{0}\left( t\right) \right\vert $ and $\rho
_{0}\left( t\right) =\frac{\bar{R}_{0}\left( t\right) }{\left\Vert \bar{R}%
_{0}\right\Vert _{L^{\infty }\left[ 0,1\right] }},$ after omitting bars,
transform system of equations (\ref{1a}) - (\ref{1c}), subject to (\ref{1d}%
), into%
\begin{gather}
\dot{v}=-\lambda _{1}y,\;\;\dot{m}=-v\cos \theta -\lambda _{2}\sin \theta
,\;\;\dot{x}=\cos \theta ,\;\;\dot{y}=\sin \theta ,\;\;m-\kappa \ddot{m}=%
\dot{\theta}-\frac{\alpha _{1}}{\rho _{0}},  \label{6} \\
x(0)=0,\;\;y(0)=0,\;\;\theta (0)=0,\;\;v(1)=\alpha _{2},\;\;m(1)=0,
\label{7}
\end{gather}%
where $\left( \cdot \right) ^{\cdot }=\frac{\mathrm{d}}{\mathrm{d}t}\left(
\cdot \right) $ and $\dot{\theta}=\frac{\ddot{y}}{\sqrt{1-\dot{y}^{2}}}.$

Parameters $\lambda _{1}$ and $\lambda _{2},$ corresponding to the angular
velocity and intensity of the horizontal force, are considered as load
parameters, while parameters $\alpha _{1}$ and $\alpha _{2},$ corresponding
to the maximal value of rod's initial curvature and intensity of the
vertical force, are considered as imperfections in shape and loading. It is
obvious from the governing system of equations (\ref{6}), subject to
boundary conditions (\ref{7}), that, for all real values of load parameters
and zero values of imperfection parameters, it admits the trivial solution%
\begin{equation*}
x_{0}=S,\;\;y_{0}=0,\;\;\theta _{0}=0,\;\;v_{0}=0,\;\;m_{0}=0.
\end{equation*}%
The critical values of load parameters $\lambda _{01}$ and $\lambda
_{02}=\lambda _{02}\left( \lambda _{01}\right) $ are found in \cite{AZ-1}
using the Euler method of adjacent equilibrium configuration, i.e., by
solving for the non-trivial solutions 
the linearized system of equations (\ref{6}), subject to (\ref{7}), with $\alpha_1=0$ and $\alpha_2=0.$ 
The present analysis will show that in the
neighborhood of critical loading values there also exists the non-trivial
solution to non-linear system of equations (\ref{6}), subject to (\ref{7}),
bifurcating from the trivial solution at the critical loading value. The
stability problem for perfect rod, i.e., initially straight rod without
vertical force acting on its tip (zero values of imperfection parameters),
will be studied in Section \ref{perfect}, while in Section \ref{imperfect} 
the study will focus on the stability problem for imperfect rod, i.e., rod having small initial
deformation with vertical force of small intensity acting on its tip
(non-zero values of imperfection parameters). 
Section \ref{calc} is devoted to numerical analysis of
the interaction curve equation, mode shapes, bifurcation equation for
perfect and imperfect rod.

\section{Bifurcation points for perfect rod \label{perfect}}

The static stability problem is considered for the perfect rod, i.e.,
rotating rod without initial deformation, loaded by the horizontal axial
force acting at its tip. The system of equations describing the equilibrium
of perfect rod is (\ref{6}), subject to boundary conditions (\ref{7}), with $%
\alpha _{1}=\alpha _{2}=0$ and it reads%
\begin{gather}
\dot{v}=-\lambda _{1}y,\;\;\dot{m}=-v\sqrt{1-\dot{y}^{2}}-\lambda _{2}\dot{y}%
,\;\;m-\kappa \ddot{m}=\frac{\ddot{y}}{\sqrt{1-\dot{y}^{2}}},  \label{9} \\
y(0)=0,\;\;\dot{y}(0)=0,\;\;v(1)=0,\;\;m(1)=0.  \label{10}
\end{gather}%
System of equations (\ref{9}), subject to (\ref{10}), can be reduced to a
single equation, represented by the action of a non-linear operator on
deflection $y$ equated with zero. The operator is obtained either as the
integro-differential operator of the second order, or as the differential
operator of the fourth order. In both cases, (\ref{9})$_{2}$ is
differentiated, (\ref{9})$_{1}$ is used, and such obtained expression is
substituted into (\ref{9})$_{3}$ yielding 
\begin{equation}
m=\left( 1+\kappa v\dot{y}-\kappa \lambda _{2}\sqrt{1-\dot{y}^{2}}\right) 
\frac{\ddot{y}}{\sqrt{1-\dot{y}^{2}}}+\kappa \lambda _{1}y\sqrt{1-\dot{y}^{2}%
}.  \label{11}
\end{equation}

The equation 
\begin{equation}
M^{2}\left( \lambda ,y\right) =0,\;\;\lambda \in \mathbb{R}^{2},\;y\in
C^{k}\left( \left[ 0,1\right] \right) ,\text{ }k\geq 2,  \label{EqM}
\end{equation}%
with the operator $M^{2},$ defined by 
\begin{equation}
M^{2}\left( \lambda ,y\right) :=\ddot{y}-\sqrt{1-\dot{y}^{2}}\frac{\lambda
_{1}\left( J_{2}y-\kappa y\sqrt{1-\dot{y}^{2}}\right) +\lambda _{2}I_{1}\dot{%
y}}{1+\kappa \lambda _{1}\dot{y}\left( I_{1}y\right) -\kappa \lambda _{2}%
\sqrt{1-\dot{y}^{2}}},  \label{opM}
\end{equation}%
where, for $z\in L^{1}\left[ 0,1\right] ,$%
\begin{equation*}
J_{2}z\left( t\right) :=\int_{t}^{1}\int_{\tau }^{1}z(\eta )\sqrt{1-\dot{z}%
^{2}\left( \tau \right) }\mathrm{d}\eta \,\mathrm{d}\tau \;\;\text{and}%
\;\;I_{1}z\left( t\right) :=\int_{t}^{1}z\left( \tau \right) \mathrm{d}\tau ,
\end{equation*}%
is obtained by integrating (\ref{9})$_{1}$ and (\ref{9})$_{2},$ taking into
account (\ref{10})$_{3}$ and (\ref{10})$_{4}$ and by substituting such
obtained expressions into (\ref{11}). Note that $M^{2}:\mathbb{R}^{2}\times
C^{k}\left( \left[ 0,1\right] \right) \rightarrow C^{k-2}\left( \left[ 0,1%
\right] \right) ,$ $k\geq 2.$ The equation (\ref{EqM}) is subject to
boundary conditions (\ref{10})$_{1}$ and (\ref{10})$_{2},$ i.e.,%
\begin{equation}
BC^{2}=\left\{ y:y\left( 0\right) =0,\;\;\dot{y}\left( 0\right) =0\right\} .
\label{bc2}
\end{equation}

The equation%
\begin{equation}
M^{4}\left( \lambda ,y\right) =0,\;\;\lambda \in \mathbb{R}^{2},\;y\in
C^{k}\left( \left[ 0,1\right] \right) ,\text{ }k\geq 4,  \label{EqM4}
\end{equation}%
with the operator $M^{4},$ defined by 
\begin{equation}
M^{4}\left( \lambda ,y\right) :=\left( \frac{\Big(\frac{\ddot{y}}{\sqrt{1-%
\dot{y}^{2}}}\Big)^{\boldsymbol{\cdot }}\left( 1-\kappa \lambda _{2}\sqrt{1-%
\dot{y}^{2}}\right) +\kappa \lambda _{1}\dot{y}\left( \sqrt{1-\dot{y}^{2}}-%
\frac{2y\ddot{y}}{\sqrt{1-\dot{y}^{2}}}\right) +\lambda _{2}\dot{y}\left(
1+\kappa \Big(\frac{\ddot{y}}{\sqrt{1-\dot{y}^{2}}}\Big)^{2}\right) }{\kappa 
\dot{y}\Big(\frac{\ddot{y}}{\sqrt{1-\dot{y}^{2}}}\Big)^{\boldsymbol{\cdot }%
}+\left( 1+\kappa \Big(\frac{\ddot{y}}{\sqrt{1-\dot{y}^{2}}}\Big)^{2}\right) 
\sqrt{1-\dot{y}^{2}}}\right) ^{\boldsymbol{\cdot }}-\lambda _{1}y,
\label{opM4}
\end{equation}%
is obtained directly from (\ref{9})$_{1},$ since the term in brackets is $v,$
obtained by differentiating (\ref{11}), with the subsequent use of (\ref{9})$%
_{2}.$ Note $M^{4}:\mathbb{R}^{2}\times C^{k}\left( \left[ 0,1\right]
\right) \rightarrow C^{k-4}\left( \left[ 0,1\right] \right) ,$ $k\geq 4.$
The equation (\ref{EqM4}) is subject to boundary conditions%
\begin{eqnarray}
BC^{4} &=&\Bigg\{y:y\left( 0\right) =0,\;\;\dot{y}\left( 0\right) =0,  \notag
\\
&&\left( 1-\kappa \lambda _{2}\sqrt{1-\dot{y}^{2}\left( 1\right) }\right) 
\frac{\ddot{y}\left( 1\right) }{\sqrt{1-\dot{y}^{2}\left( 1\right) }}+\kappa
\lambda _{1}y(1)\sqrt{1-\dot{y}^{2}\left( 1\right) }%
\begin{tabular}{l}
=%
\end{tabular}%
0,  \notag \\
&&\left( 1-\kappa \lambda _{2}\sqrt{1-\dot{y}^{2}\left( 1\right) }\right)
\left. \left( \frac{\ddot{y}\left( t\right) }{\sqrt{1-\dot{y}^{2}\left(
t\right) }}\right) ^{\boldsymbol{\cdot }}\right\vert _{t=1}  \notag \\
&&+\left( \kappa \lambda _{1}\left( \sqrt{1-\dot{y}^{2}\left( 1\right) }%
-2y\left( 1\right) \frac{\ddot{y}\left( 1\right) }{\sqrt{1-\dot{y}^{2}\left(
1\right) }}\right) +\lambda _{2}\left( 1+\kappa \left( \frac{\ddot{y}\left(
1\right) }{\sqrt{1-\dot{y}^{2}\left( 1\right) }}\right) ^{2}\right) \right) 
\dot{y}(1)%
\begin{tabular}{l}
=%
\end{tabular}%
0\Bigg\},  \label{bc4}
\end{eqnarray}%
where the first two boundary conditions are (\ref{10})$_{1}$ and (\ref{10})$%
_{2},$ while the third boundary condition is (\ref{10})$_{4},$ with (\ref{11}%
) calculated at $t=1$ and the fourth boundary condition is (\ref{10})$_{3},$
with the nominator of the term in brackets in (\ref{opM4}) calculated at $%
t=1.$

Equations (\ref{EqM}), subject to (\ref{bc2}), and (\ref{EqM4}), subject to (%
\ref{bc4}), are equivalent. The focus is on finding bifurcation points to
problem (\ref{EqM}), (\ref{bc2}) (or equivalently to (\ref{EqM4}), (\ref{bc4}%
)). It is easy to verify that for all $\lambda \in \mathbb{R}^{2}$ there is
a solution curve of (\ref{EqM}), (\ref{bc2}) (and of (\ref{EqM4}), (\ref{bc4}%
)), through $\left( \lambda ,0\right) $ and the critical value $\lambda _{0}$
for which there are other solution curves in neighborhood $U\times V\subset 
\mathbb{R}^{2}\times C^{k}\left( \left[ 0,1\right] \right) ,$ $k\geq 2,$ of $%
\left( \lambda _{0},0\right) $ for problem (\ref{EqM}), (\ref{bc2}) (or in
neighborhood $U\times V\subset \mathbb{R}^{2}\times C^{k}\left( \left[ 0,1%
\right] \right) ,$ $k\geq 4,$ of $\left( \lambda _{0},0\right) $ for problem
(\ref{EqM4}), (\ref{bc4})) are sought for. A necessary condition for $%
\lambda _{0}$ to be critical value is the failure of implicit function
theorem, see e.g. \cite[Theorem I.1.1]{Kielhofer}, i.e., that 
\begin{equation}
D_{y}M^{j}\left( \lambda _{0},0\right) :C^{k}\left( \left[ 0,1\right]
\right) \mapsto C^{k-j}\left( \left[ 0,1\right] \right) \;\text{is not
bijective,}  \label{uslov}
\end{equation}%
with $j\in \left\{ 2,4\right\} $ and $k\geq 2$ for $j=2$ and $k\geq 4$ for $%
j=4,$ where $D_{y}$ denotes the Fr\'{e}chet derivative. The Fr\'{e}chet
derivatives of $M^{2}$, $M^{4}$ and $BC^{4}$ at $\left( \lambda ,0\right) $
are calculated as%
\begin{eqnarray}
&&\!\!\!\!\!\!\!\!\!\!\!\!\!\!\!\!\!\!L^{2}\left( \lambda \right) y\!\!%
\begin{tabular}{l}
:=%
\end{tabular}%
\!\!D_{y}M^{2}\left( \lambda ,0\right) y%
\begin{tabular}{l}
=%
\end{tabular}%
\ddot{y}-\frac{\lambda _{1}}{1-\kappa \lambda _{2}}\left( I_{2}y-\kappa
y\right) -\frac{\lambda _{2}}{1-\kappa \lambda _{2}}I_{1}\dot{y}  \label{opL}
\\
&&\!\!\!\!\!\!\!\!\!\!\!\!\!\!\!\!\!\!L^{4}\left( \lambda \right) y\!\!%
\begin{tabular}{l}
:=%
\end{tabular}%
\!\!D_{y}M^{4}\left( \lambda ,0\right) y%
\begin{tabular}{l}
=%
\end{tabular}%
y^{\mathrm{IV}}(t)+\frac{\kappa \lambda _{1}+\lambda _{2}}{1-\kappa \lambda
_{2}}\ddot{y}\left( t\right) -\frac{\lambda _{1}}{1-\kappa \lambda _{2}}y(t),
\label{opL4} \\
&&\!\!\!\!\!\!\!\!\!\!\!\!\!\!\!\!\!\!LBC\!\!%
\begin{tabular}{l}
=%
\end{tabular}%
\!\!\left\{ y\left( 0\right) =0,\;\dot{y}\left( 0\right) =0,\;\ddot{y}\left(
1\right) \left( 1-\kappa \lambda _{2}\right) +\kappa \lambda _{1}y(1)=0,\;y^{%
\mathrm{III}}\left( 1\right) \left( 1-\kappa \lambda _{2}\right) +\left(
\kappa \lambda _{1}+\lambda _{2}\right) \dot{y}(1)=0\right\} ,  \label{LBC}
\end{eqnarray}%
where%
\begin{equation*}
I_{2}z\left( t\right) :=\int_{t}^{1}\int_{\tau }^{1}z\left( \eta \right) 
\mathrm{d}\eta \mathrm{d}\tau .
\end{equation*}

Finding $\lambda _{0}$ such that (\ref{uslov}) holds is equivalent to
finding $\lambda _{0}$ for which kernel of the operator $L^{2}\left( \lambda
_{0}\right) $ (or $L^{4}\left( \lambda _{0}\right) $) is nontrivial (do not
consists of $y=0$ only). For fixed $\lambda $, one finds kernel of the
operator $L^{2}\left( \lambda \right) $ (or $L^{4}\left( \lambda \right) $)
by solving the equation 
\begin{equation}
L^{j}\left( \lambda \right) y=0,\;\;y\in Y^{j},\;j\in \left\{ 2,4\right\} ,
\label{EqL}
\end{equation}%
where%
\begin{equation*}
Y^{2}=\left\{ y:y\in C^{k}\left( \left[ 0,1\right] \right) ,\;k\geq
2\right\} \cap BC^{2}\;\;\text{and}\;\;Y^{4}=\left\{ y:y\in C^{k}\left( %
\left[ 0,1\right] \right) ,\;k\geq 4\right\} \cap LBC,
\end{equation*}%
where $BC^{2}$ and $LBC$ are given by (\ref{bc2}) and (\ref{LBC}). Note that 
$Y^{2}$ and $Y^{4}$ are Hilbert spaces with usual scalar product $%
\left\langle y,q\right\rangle =\int_{0}^{1}y\left( t\right) q\left( t\right) 
\mathrm{d}t.$

The problems (\ref{EqL}) for $j=2$ and (\ref{EqL}) for $j=4$ are equivalent.
Indeed, $\frac{\mathrm{d}^{2}}{\mathrm{d}t^{2}}\left( L^{2}\left( \lambda
\right) y\right) =L^{4}\left( \lambda \right) y,$ with boundary conditions (%
\ref{LBC}) obtained for $L^{2}\left( \lambda \right) y\left( 1\right) =0$
and $\left. \frac{\mathrm{d}}{\mathrm{d}t}\left( L^{2}\left( \lambda \right)
y\left( t\right) \right) \right\vert _{t=1}=0,$ while $I_{2}\left(
L^{4}\left( \lambda \right) y\right) =L^{2}\left( \lambda \right) y$ is
obtained by integration of (\ref{opL4}) and use of the boundary conditions (%
\ref{LBC}).

The problem (\ref{EqL}) for $j=4$ is considered in \cite{AZ-1}. The critical
value $\lambda _{0}=\left( \text{$\lambda _{01},\lambda _{02}$}\right) $ is
obtained from the condition of existence of nontrivial solution $y$ to
problem (\ref{EqL}), $j=4$, which requires that the determinant arising from
boundary conditions (\ref{LBC}) is equal to zero, i.e., as a solution of 
\begin{eqnarray}
&&f\left( \text{$\lambda _{1},\lambda _{2}$}\right) =\sqrt{\frac{\text{$%
\lambda _{1}$}}{1-\kappa \text{$\lambda _{2}$}}}\Bigg(2\text{$\lambda _{1}$}%
+\kappa \text{$\lambda _{1}$}\left( \kappa \text{$\lambda _{1}$}-\text{$%
\lambda _{2}$}\right) +\left( 2\text{$\lambda _{1}$}+\text{$\lambda _{2}^{2}$%
}-\kappa \text{$\lambda _{1}\lambda _{2}$}\right) \cos \left( r_{1}\left( 
\text{$\lambda _{1},\lambda _{2}$}\right) \right) \,\cosh \left( r_{2}\left( 
\text{$\lambda _{1},\lambda _{2}$}\right) \right)  \notag \\
&&\qquad \qquad \qquad \qquad -\sqrt{\frac{\text{$\lambda _{1}$}}{1-\kappa 
\text{$\lambda _{2}$}}}\left( \text{$\lambda _{2}$}-\kappa \left( \text{$%
\lambda _{1}$}-\kappa \text{$\lambda _{1}\lambda _{2}$}+\text{$\lambda
_{2}^{2}$}\right) \right) \sin \left( r_{1}\left( \text{$\lambda
_{1},\lambda _{2}$}\right) \right) \,\sinh \left( r_{2}\left( \text{$\lambda
_{1},\lambda _{2}$}\right) \right) \Bigg)=0,  \label{frekventna}
\end{eqnarray}%
where%
\begin{eqnarray}
r_{1}\left( \text{$\lambda _{1},\lambda _{2}$}\right) &=&\sqrt{\sqrt{\frac{%
\lambda _{1}}{1-\kappa \lambda _{2}}+\left( \frac{1}{2}\frac{\kappa \lambda
_{1}+\lambda _{2}}{1-\kappa \lambda _{2}}\right) ^{2}}+\frac{1}{2}\frac{%
\kappa \lambda _{1}+\lambda _{2}}{1-\kappa \lambda _{2}}},
\label{frekventna1} \\
r_{2}\left( \text{$\lambda _{1},\lambda _{2}$}\right) &=&\sqrt{\sqrt{\frac{%
\lambda _{1}}{1-\kappa \lambda _{2}}+\left( \frac{1}{2}\frac{\kappa \lambda
_{1}+\lambda _{2}}{1-\kappa \lambda _{2}}\right) ^{2}}-\frac{1}{2}\frac{%
\kappa \lambda _{1}+\lambda _{2}}{1-\kappa \lambda _{2}}}.
\label{frekventna2}
\end{eqnarray}

By the implicit function theorem, since $f\left( \text{$\lambda
_{01},\lambda _{02}$}\right) =0$ and $\left. \frac{\partial f\left( \text{$%
\lambda _{1},\lambda _{2}$}\right) }{\partial \lambda _{2}}\right\vert
_{\left( \text{$\lambda _{1},\lambda _{2}$}\right) =\left( \lambda
_{01},\lambda _{02}\right) }\neq 0,$ in the neighborhood of $\lambda
_{01},\lambda _{02},$ i.e., for $\lambda _{1}=\lambda _{01}+\Delta \lambda
_{1}$ and $\lambda _{2}=\lambda _{02}+\Delta \lambda _{2},$ equation (\ref%
{frekventna}) is solved with respect to $\lambda _{2},$ i.e., there exists a
unique differentiable function $\eta ,$ such that $\lambda _{2}=\eta \left( 
\text{$\lambda _{1}$}\right) $ and 
\begin{equation}
f\left( \text{$\lambda _{1},$}\eta \left( \text{$\lambda _{1}$}\right)
\right) =0\;\;\text{and}\;\;\eta ^{\prime }\left( \text{$\lambda _{1}$}%
\right) =\frac{\mathrm{d}\eta \left( \text{$\lambda _{1}$}\right) }{\mathrm{d%
}\text{$\lambda _{1}$}}=-\left. \frac{\frac{\partial f\left( \text{$\lambda
_{1},\lambda _{2}$}\right) }{\partial \lambda _{1}}}{\frac{\partial f\left( 
\text{$\lambda _{1},\lambda _{2}$}\right) }{\partial \lambda _{2}}}%
\right\vert _{\left( \text{$\lambda _{1},\lambda _{2}$}\right) =\left( \text{%
$\lambda _{1},$}\eta \left( \text{$\lambda _{1}$}\right) \right) }.
\label{eta prim}
\end{equation}%
Then also 
\begin{equation}
\Delta \lambda _{2}=\eta ^{\prime }\left( \text{$\lambda _{01}$}\right)
\Delta \lambda _{1}.  \label{delta lambda}
\end{equation}

Nontrivial solution to (\ref{EqL}), $j=4$, corresponding to $\lambda
_{0}=\left( \lambda _{01},\lambda _{02}\right) $ reads: 
\begin{equation}
y_{l}\left( t\right) =C\left( \cos \left( r_{01}t\right) -\cosh \left(
r_{02}t\right) -D\left( r_{01},r_{02}\right) \left( \sin \left(
r_{01}t\right) -\frac{r_{01}}{r_{02}}\sinh \left( r_{02}t\right) \right)
\right) ,  \label{Linearno}
\end{equation}%
where $C$ is an arbitrary constant and $D$ is a constant given by 
\begin{equation*}
D\left( r_{01},r_{02}\right) =\frac{r_{01}^{2}\cos r_{01}+r_{02}^{2}\cosh
r_{02}+\frac{\kappa \text{$\lambda _{01}$}}{1-\kappa \text{$\lambda _{02}$}}%
\left( \cosh r_{02}-\cos r_{01}\right) }{r_{01}^{2}\sin
r_{01}+r_{01}r_{02}\sinh r_{02}+\frac{\kappa \text{$\lambda _{01}$}}{%
1-\kappa \text{$\lambda _{02}$}}\left( \frac{r_{01}}{r_{02}}\sinh
r_{02}-\sin r_{01}\right) },
\end{equation*}%
where parameters $r_{01}$ and $r_{02}$ are calculated from (\ref{frekventna1}%
) and (\ref{frekventna2}) for $\lambda _{0}$.

The kernel of operator $L^{j}\left( \lambda _{0}\right) ,$ $j\in \left\{
2,4\right\} ,$ is one-dimensional space, i.e., 
\begin{equation}
\dim N\left( L^{j}\left( \lambda _{0}\right) \right) =1,\;\;j\in \left\{
2,4\right\} ,  \label{dimN}
\end{equation}%
since $N\left( L^{j}\left( \lambda _{0}\right) \right) =\func{span}%
[y_{L}]=\{ay_{L};\,a\in \mathbb{R}\},$ $j\in \left\{ 2,4\right\} ,$ where
the normalized solution (\ref{Linearno}) is denoted by $y_{L}$, i.e. the
solution with constant $C$ chosen such that $\left\Vert y_{L}\right\Vert
_{Y^{j}}=1.$

Orthogonal complement of the range of $L^{j}\left( \lambda _{0}\right) $ is
a kernel of the formal adjoint $L^{j\ast }\left( \lambda _{0}\right) $ of
operator $L^{j}\left( \lambda _{0}\right) $, where the formal adjoint of an
operator $L^{j}:Y^{j}\rightarrow Z^{j}$ is defined as an operator $L^{j\ast
}:Z^{j}\rightarrow Y^{j},$ such that for all $y\in Y^{j}$ and all $q\in
Z^{j} $ equality $\langle L^{j\ast }q,y\rangle _{Y^{j}}=\langle
L^{j}y,q\rangle _{Z^{j}}$ holds, where $Z^{j}=C^{k-j}\left( \left[ 0,1\right]
\right) ,$ $j\in \left\{ 2,4\right\} .$ Straightforward calculation gives%
\begin{eqnarray*}
&&L^{2\ast }\left( \lambda \right) q%
\begin{tabular}{l}
=%
\end{tabular}%
\ddot{y}-\frac{\lambda _{01}}{1-\kappa \lambda _{02}}\left( I_{2}y-\kappa
y\right) -\frac{\lambda _{02}}{1-\kappa \lambda _{02}}I_{1}\dot{y}, \\
&&LBC^{2\ast }%
\begin{tabular}{l}
=%
\end{tabular}%
\bigg\{q(1)=0,\;\;\dot{q}(1)+\frac{\lambda _{02}}{1-\kappa \lambda _{02}}%
\left\langle 1,q\right\rangle =0,\;\;\ddot{q}\left( 1\right) +\frac{\lambda
_{01}}{1-\kappa \lambda _{02}}\left( \left\langle t,q\right\rangle
-\left\langle 1,q\right\rangle \right) =0 \\
&&\qquad \qquad \qquad q^{\mathrm{III}}\left( 1\right) +\frac{\kappa \lambda
_{01}+\lambda _{02}}{1-\kappa \lambda _{02}}\dot{q}(1)-\frac{\lambda _{01}}{%
1-\kappa \lambda _{02}}\left\langle 1,q\right\rangle =0\bigg\}, \\
&&L^{4\ast }\left( \lambda \right) q%
\begin{tabular}{l}
=%
\end{tabular}%
q^{\mathrm{IV}}(t)+\frac{\kappa \lambda _{01}+\lambda _{02}}{1-\kappa
\lambda _{02}}\ddot{q}\left( t\right) -\frac{\lambda _{01}}{1-\kappa \lambda
_{02}}q(t), \\
&&LBC^{4\ast }%
\begin{tabular}{l}
=%
\end{tabular}%
\left\{ q(0)=0,\;\;\dot{q}(0)=0,\;\;\ddot{q}\left( 1\right) =0,\;\;q^{%
\mathrm{III}}\left( 1\right) +\frac{\lambda _{02}}{1-\kappa \lambda _{02}}%
\dot{q}(1)=0\right\} .
\end{eqnarray*}%
The kernel of operator $L^{j\ast }\left( \lambda _{0}\right) $ is found by
solving equation $L^{j\ast }\left( \lambda _{0}\right) q=0$, $q\in Z^{j},$
whose solution reads%
\begin{eqnarray}
q_{l}^{\left( 2\right) }\left( t\right) &=&C\left( \cos \left(
r_{01}t\right) +\frac{r_{02}^{2}}{r_{01}^{2}}\cosh \left( r_{02}t\right) -%
\frac{\cos r_{01}+\frac{r_{02}^{2}}{r_{01}^{2}}\cosh r_{02}}{\sin r_{01}+%
\frac{r_{02}}{r_{01}}\sinh r_{02}}\left( \sin \left( r_{01}t\right) +\frac{%
r_{02}}{r_{01}}\sinh (r_{02}t)\right) \right) ,  \label{q2-l} \\
q_{l}^{\left( 4\right) }\left( t\right) &=&C\left( \cos \left(
r_{01}t\right) -\cosh \left( r_{02}t\right) -\frac{\cos r_{01}+\frac{%
r_{02}^{2}}{r_{01}^{2}}\cosh r_{02}}{\sin r_{01}+\frac{r_{02}}{r_{01}}\sinh
r_{02}}\left( \sin \left( r_{01}t\right) -\frac{r_{01}}{r_{02}}\sinh
(r_{02}t)\right) \right) .  \label{q4-l}
\end{eqnarray}%
Therefore, the kernel of operator $L^{j\ast }\left( \lambda _{0}\right) $ is
one-dimensional and 
\begin{equation}
\func{codim}R(L^{j}\left( \lambda _{0}\right) )=1,\;\;j\in \left\{
2,4\right\} .  \label{codimR}
\end{equation}

If $\lambda _{0}=\left( \lambda _{01},\eta \left( \lambda _{01}\right)
\right) $ is critical value, then by Krasnoselskii theorem, $%
\left( \lambda _{0},0\right) $ is a bifurcation point of the nonlinear
operators $M^{2}$ and $M^{4},$ since$,$ according to (\ref{dimN}), $\dim
N\left( L^{j}\left( \lambda _{0}\right) \right) =1$ and it is of odd
algebraic multiplicity. Although $\left( \lambda _{0},0\right) $ is proved
to be a bifurcation point, the existence of nontrivial solution to (\ref%
{EqM4}) is also established by the use of Crandall-Rabinowitz theorem, see 
\cite[Theorem I.5.1]{Kielhofer}.

\begin{theorem}
\label{Crandal}Let $Y^{4}$ and $Z^{4}$ be defined as above and let operator $%
M^{4}$ be given by (\ref{opM4}). Let $\lambda _{0}=\left( \lambda _{01},\eta
\left( \lambda _{01}\right) \right) $ be the critical value for which there
exists nontrivial solution to (\ref{EqL}). Then $\left( \lambda
_{0},0\right) =\left( \lambda _{01},\eta \left( \lambda _{01}\right)
,0\right) $ is a bifurcation point to (\ref{EqM4}).
\end{theorem}

\begin{proof}
Let $U$ and $V$ be open neighborhoods in $\mathbb{R}$ and $Y^{4}$ such that $%
\lambda _{01}\in U\subset \mathbb{R}$ and $0\in V\subset Y^{4}.$ Let $\bar{M}%
^{4}$ be operator on $\mathbb{R}\times Y^{4}$ defined as 
\begin{equation*}
\bar{M}^{4}\left( \lambda _{1},y\right) :=M^{4}\left( \lambda _{1},\eta
\left( \lambda _{1}\right) ,y\right) .
\end{equation*}%
Note that $\bar{M}^{4}\in C^{2}\left( U\times V,Z^{4}\right) $ and that $%
\bar{M}^{4}\left( \lambda _{1},0\right) =0\;$for\ all $\lambda _{1}\in 
\mathbb{R}.$ According to (\ref{dimN}) and (\ref{codimR}), the operator $%
\bar{M}^{4}\left( \lambda _{01},\cdot \right) $ is Fredholm operator of
index zero. Further, by showing that $D_{y,\lambda _{1}}^{2}\bar{M}%
^{4}\left( \lambda _{01},0\right) y_{L}$ belongs to $N(L^{4\ast }\left(
\lambda _{0}\right) ),$ 
it will be proved that the requirement Crandall-Rabinowitz theorem \cite[Theorem I.5.1]{Kielhofer}
$D_{y,\lambda _{1}}^{2}\bar{M}%
^{4}\left( \lambda _{01},0\right) y_{L}\notin R\left( D_{y}\bar{M}^{4}\left(
\lambda _{01},0\right) \right)$ is satisfied.
Indeed,
\begin{equation*}
D_{y,\lambda _{1}}^{2}\bar{M}^{4}\left( \lambda _{01},0\right) y_{L}=\left(
\kappa \Lambda _{1}+\Lambda _{2}\right) \ddot{y}_{L}\left( t\right) -\Lambda
_{1}y_{L}(t),
\end{equation*}%
where 
\begin{equation*}
\Lambda _{1}=\frac{\kappa \lambda _{01}\eta ^{\prime }\left( \lambda
_{01}\right) -\kappa \eta \left( \lambda _{01}\right) +1}{\left( 1-\kappa
\eta \left( \lambda _{01}\right) \right) ^{2}},\;\text{\ }\Lambda _{2}=\frac{%
\eta ^{\prime }\left( \lambda _{01}\right) }{\left( 1-\kappa \eta \left(
\lambda _{01}\right) \right) ^{2}},
\end{equation*}%
so that 
\begin{align*}
L^{4\ast }\left( \lambda _{0}\right) & D_{y,\lambda _{1}}^{2}\bar{M}%
^{4}\left( \lambda _{01},0\right) y_{L}=L^{4\ast }\left( \lambda _{0}\right) %
\Big(\left( \kappa \Lambda _{1}+\Lambda _{2}\right) \ddot{y}_{L}\left(
t\right) -\Lambda _{1}y_{L}(t)\Big) \\
& =\Big(\left( \kappa \Lambda _{1}+\Lambda _{2}\right) \ddot{y}_{L}\left(
t\right) -\Lambda _{1}y_{L}(t)\Big)^{\mathrm{IV}}+\frac{\kappa \lambda
_{01}+\eta \left( \lambda _{01}\right) }{1-\kappa \eta \left( \lambda
_{01}\right) }\Big(\left( \kappa \Lambda _{1}+\Lambda _{2}\right) \ddot{y}%
_{L}\left( t\right) -\Lambda _{1}y_{L}\left( t\right) \Big)^{\boldsymbol{%
\cdot }\boldsymbol{\cdot }} \\
& \qquad -\frac{\lambda _{01}}{1-\kappa \eta \left( \lambda _{01}\right) }%
\Big(\left( \kappa \Lambda _{1}+\Lambda _{2}\right) \ddot{y}_{L}\left(
t\right) -\Lambda _{1}y_{L}(t)\Big) \\
& =\left( \kappa \Lambda _{1}+\Lambda _{2}\right) \left( y_{L}^{\mathrm{IV}%
}(t)+\frac{\kappa \lambda _{01}+\eta \left( \lambda _{01}\right) }{1-\kappa
\eta \left( \lambda _{01}\right) }\ddot{y}_{L}\left( t\right) -\frac{\lambda
_{01}}{1-\kappa \eta \left( \lambda _{01}\right) }y_{L}(t)\right) ^{%
\boldsymbol{\cdot }\boldsymbol{\cdot }} \\
& \qquad -\Lambda _{1}\left( y_{L}^{\mathrm{IV}}(t)+\frac{\kappa \lambda
_{01}+\eta \left( \lambda _{01}\right) }{1-\kappa \eta \left( \lambda
_{01}\right) }\ddot{y}_{L}\left( t\right) -\frac{\lambda _{01}}{1-\kappa
\eta \left( \lambda _{01}\right) }y_{L}(t)\right) 
\begin{tabular}{l}
=%
\end{tabular}%
0.
\end{align*}%
Thus, $\left( \lambda _{0},0\right) $ is the bifurcation point.
\end{proof}

In order to determine the type of bifurcation at point $\left( \lambda
_{0},0\right) ,$ the reduction method of Lyapunov-Schmidt will be used. Let $%
Y^{2}$ and $Z^{2}$ be defined as above and let operator $M^{2}$ be given by (%
\ref{opM}). Consider mapping $M^{2}:U\times V\rightarrow Z^{2}$, with $U$
and $V$ being open neighborhoods of $\lambda =\lambda _{0}$ and $y=0$,
respectively. According to Definition I.2.1 in \cite{Kielhofer}, (\ref{dimN}%
), and (\ref{codimR}), the operator $M^{2}(\lambda _{0},\cdot ):V\rightarrow
Z^{2}$ is a nonlinear Fredholm operator and there exist closed complements
in the Hilbert spaces $Y^{2}$ and $Z^{2}$ such that 
\begin{eqnarray}
Y^{2} &=&N\left( L^{2}\left( \lambda _{0}\right) \right) \oplus N^{\perp
}\left( L^{2}\left( \lambda _{0}\right) \right) ,  \label{splitY} \\
Z^{2} &=&R\left( L^{2}\left( \lambda _{0}\right) \right) \oplus R^{\perp
}\left( L^{2}\left( \lambda _{0}\right) \right) =R\left( L^{2}\left( \lambda
_{0}\right) \right) \oplus N\left( L^{2\ast }\left( \lambda _{0}\right)
\right) ,  \notag
\end{eqnarray}%
and there are continuous projectors 
\begin{eqnarray}
P &:&Y^{2}\rightarrow N\left( L^{2}\left( \lambda _{0}\right) \right) \;\;%
\text{and}\;\;\left( I-P\right) :Y^{2}\rightarrow N^{\perp }\left(
L^{2}\left( \lambda _{0}\right) \right) =R\left( L^{2}\left( \lambda
_{0}\right) \right) ,  \notag \\
Q &:&Z^{2}\rightarrow R^{\perp }\left( L^{2}\left( \lambda _{0}\right)
\right) =N\left( L^{2\ast }\left( \lambda _{0}\right) \right) \;\;\text{and}%
\;\;\left( I-Q\right) :Z^{2}\rightarrow R\left( L^{2}\left( \lambda
_{0}\right) \right) .  \label{projQ}
\end{eqnarray}

\begin{theorem}
\label{Ljapunov-Smit}Let $\left( \lambda _{0},0\right) $ be bifurcation
point obtained in Theorem \ref{Crandal}. Let $c_{11},$ $c_{12},$ and $c_{3}$
be given by (\ref{ce11}), (\ref{ce12}), and (\ref{ce3}), respectively. If $%
c_{3}\neq 0$ and $c_{11}+c_{12}\eta ^{\prime }\left( \lambda _{01}\right)
\neq 0,$ where $\eta $ is defined as above, then problem (\ref{EqM}),
subject to (\ref{bc2}), can be reduced to a bifurcation equation $\phi
(a,\lambda )=0$, given by (\ref{BifJedF}), which is strongly equivalent to
equation%
\begin{equation*}
\varepsilon a^{3}+\delta \Delta \lambda _{1}a=0,\;\;\text{with}%
\;\;\varepsilon =\mathrm{sgn\,}c_{3},\;\delta =\mathrm{sgn\,}\left(
c_{11}+c_{12}\eta ^{\prime }\left( \text{$\lambda _{01}$}\right) \right) ,
\end{equation*}%
i.e., problem (\ref{EqM}), (\ref{bc2}) has a pitchfork bifurcation.
\end{theorem}

\begin{proof}
Following the standard procedure \cite%
{ChowHale,GolubitskySchaeffer,Kielhofer}, equation (\ref{EqM}) is rewritten
as%
\begin{eqnarray}
QM^{2}\left( \lambda ,y\right) &=&0,  \label{BifJed} \\
\left( I-Q\right) M^{2}\left( \lambda ,y\right) &=&0,  \label{IFTjed}
\end{eqnarray}%
where $Q$ is projector defined by (\ref{projQ}). First, equation (\ref%
{IFTjed}) is solved and then its solution is inserted into (\ref{BifJed}) to
obtain bifurcation equation which will yield pitchfork bifurcation.

Due to splitting in (\ref{splitY}), function $y\in Y^{2}$ can be written as $%
y=ay_{L}+w,$ $a\in \mathbb{R},$ where $y_{L}\in N(L^{2}(\lambda _{0}))$ is
normalized solution (\ref{Linearno}) and $w\in N^{\perp }(L^{2}(\lambda
_{0})).$ Solvability of equation (\ref{IFTjed}), depending on $\lambda $, $%
ay_{L}$ and $w,$ with respect to $w$ is considered in the neighborhood of $%
(\lambda _{0},0).$ Since 
\begin{equation*}
\left( I-Q\right) D_{w}M^{2}\left( \lambda _{0},0+0\right) =\left(
I-Q\right) D_{y}M^{2}\left( \lambda _{0},0\right) =\left( I-Q\right)
L^{2}(\lambda _{0})=L^{2}(\lambda _{0})
\end{equation*}%
is invertible when considered as mapping from $N^{\perp }(L^{2}\left(
\lambda _{0}\right) )$ to $R\left( L^{2}\left( \lambda _{0}\right) \right) $%
), using the implicit function theorem a $C^{2}$ function $w=w\left( \lambda
,ay_{L}\right) $, defined in a neighborhood $U\times V\subset \mathbb{R}%
\times N(L^{2}(\lambda _{0}))$ of $(\lambda _{0},0)$, i.e. $w:U\times
V\rightarrow N^{\perp }(L^{2}\left( \lambda _{0}\right) )\subset Y^{2}$,
such that 
\begin{equation*}
\left( I-Q\right) M^{2}\left( \lambda ,ay_{L}+w(\lambda ,ay_{L})\right)
=0,\;\;\lambda \in U,\;ay_{L}\in V,
\end{equation*}%
is found.

For the later use note that $w=O\left( \left\vert ay_{L}\right\vert
^{2}\right) =O(a^{2})$ (since $w\in N^{\perp }(L^{2}(\lambda _{0}))$) and
even more%
\begin{equation*}
w=O\left( \left\vert ay_{L}\right\vert ^{3}\right) =O(a^{3}),
\end{equation*}%
since $w$ is antisymmetric with respect to $y.$ Indeed, since the operator $%
M^{2}$ is antisymmetric with respect to $y,$ ($M^{2}\left( \lambda ,y\right)
=-M^{2}\left( \lambda ,-y\right) $), one can see that $w^{\ast }=-w\left(
\lambda ,-ay_{L}\right) $ is also solution to (\ref{IFTjed}) and since, by
the implicit function theorem, solution to (\ref{IFTjed}) is unique, the
equality $w=w^{\ast }$ holds, i.e., $w\left( \lambda ,ay_{L}\right)
=-w\left( \lambda ,-ay_{L}\right) $.

Further, function $y=ay_{L}+w(\lambda ,ay_{L}),$ $a\in 
%TCIMACRO{\U{211d} }%
%BeginExpansion
\mathbb{R}
%EndExpansion
,$ $\lambda =\lambda _{0}+\Delta \lambda $ with $\left\vert \Delta \lambda
\right\vert \ll 1,$ which is a solution to (\ref{IFTjed}) in a neighborhood
of $(\lambda _{0},0),$ is a solution to (\ref{EqM}), (\ref{bc2}) if and only
if $(\Delta \lambda ,a)$ satisfies bifurcation equation $\phi (\Delta
\lambda ,a)=0$, with $\phi $ given by (\ref{BifJedF}) below which is
obtained as follows.

Rewriting the operators $M^{2}$ and $L^{2},$ given by (\ref{opM}) and (\ref%
{opL}), as%
\begin{equation*}
M^{2}\left( \lambda ,y\right) =L^{2}\left( \lambda \right) y+N\left( \lambda
,y\right) \;\;\text{and}\;\;L^{2}\left( \lambda \right) =L^{2}\left( \lambda
_{0}\right) +\tilde{L}\left( \lambda \right) ,
\end{equation*}%
with $\lambda _{0}$ being the critical value for which there exists
nontrivial solution to (\ref{EqL}) 
\begin{align}
N\left( \lambda ,y\right) & :=M^{2}\left( \lambda ,y\right) -L^{2}\left(
\lambda \right) y  \notag \\
& =-\sqrt{1-\dot{y}^{2}}\frac{\lambda _{1}\left( J_{2}y-\kappa y\sqrt{1-\dot{%
y}^{2}}\right) +\lambda _{2}I_{1}\dot{y}}{1+\kappa \lambda _{1}\dot{y}\left(
I_{1}y\right) -\kappa \lambda _{2}\sqrt{1-\dot{y}^{2}}}+\frac{\lambda
_{1}\left( I_{2}y-\kappa y\right) +\lambda _{2}I_{1}\dot{y}}{1-\kappa
\lambda _{2}}  \label{N}
\end{align}%
and 
\begin{align*}
\tilde{L}\left( \lambda \right) y &:=L^{2}\left( \lambda \right)
y-L^{2}\left( \lambda _{0}\right) y \\
&=-\frac{\lambda _{1}\left( I_{2}y-\kappa y\right) +\lambda _{2}I_{1}\dot{y}%
}{1-\kappa \lambda _{2}}+\frac{\lambda _{01}\left( I_{2}y-\kappa y\right)
+\lambda _{02}I_{1}\dot{y}}{1-\kappa \lambda _{02}}
\end{align*}%
are obtained.

Since $QL\left( \lambda _{0}\right) y=0,$ equation (\ref{BifJed}) is
equivalent to 
\begin{equation}
Q\left( \tilde{L}\left( \lambda \right) y+N\left( \lambda ,y\right) \right)
=0.  \label{BifJed1}
\end{equation}%
Note that $Q\left( \tilde{L}\left( \lambda \right) y+N\left( \lambda
,y\right) \right) \in N\left( L^{2\ast }\left( \lambda _{0}\right) \right) $
and that for all $q\in N\left( L^{2\ast }\left( \lambda _{0}\right) \right) $
it holds 
\begin{align*}
\left\langle \tilde{L}\left( \lambda \right) y+N\left( \lambda ,y\right)
,q\right\rangle & =\left\langle Q\left( \tilde{L}\left( \lambda \right)
y+N\left( \lambda ,y\right) \right) ,q\right\rangle +\left\langle \left(
I-Q\right) \left( \tilde{L}\left( \lambda \right) y+N\left( \lambda
,y\right) \right) ,q\right\rangle \\
& =\left\langle Q\left( \tilde{L}\left( \lambda \right) y+N\left( \lambda
,y\right) \right) ,q\right\rangle ,
\end{align*}%
so for (\ref{BifJed1}), as well as for (\ref{BifJed}), to hold, it is
sufficient and necessary that for all $q\in N\left( L^{2\ast }\left( \lambda
_{0}\right) \right) $ 
\begin{equation}
\left\langle \tilde{L}\left( \lambda \right) y+N\left( \lambda ,y\right)
,q\right\rangle =0.  \label{BifJed2}
\end{equation}

Taylor's expansions of operators $N$ and $\tilde{L}$ are calculated in a
neighborhood of $(\lambda _{0},0),$ i.e., for $y=ay_{L}+w(\lambda ,ay_{L}),$ 
$w=O(a^{3}),$ and $\lambda =\lambda _{0}+\Delta \lambda $, $\left\vert
\Delta \lambda \right\vert \ll 1$, in two steps. In the first step, operator 
$N$ is expanded up to third order with respect to $y.$ In the second step, $%
y=ay_{L}+O(a^{3})$ and $\lambda =\lambda _{0}+\Delta \lambda $, $\left\vert
\Delta \lambda \right\vert \ll 1,$ are put in the expression for $N$
obtained in the first step and in $\tilde{L}$.

In the first step, due to $y=ay_{L}+O(a^{3}),$ $y=O\left( a\right) ,$ so the
operator $N$ takes the form%
\begin{eqnarray*}
N\left( \lambda ,y\right) \!\!\!\! &=&\!\!\!\!\frac{1}{2\left( 1-\kappa
\lambda _{2}\right) }\bigg(\lambda _{1}\Big(I_{3}y+\dot{y}^{2}\left(
I_{2}y-2\kappa y\right) \Big)+\lambda _{2}\dot{y}^{2}\left( I_{1}\dot{y}%
\right) \bigg) \\
&&\!\!\!\!+\frac{\kappa }{2\left( 1-\kappa \lambda _{2}\right) ^{2}}\bigg(%
2\lambda _{1}^{2}\dot{y}\left( I_{1}y\right) \left( I_{2}y-\kappa y\right)
+\lambda _{1}\lambda _{2}\dot{y}\Big(\dot{y}\left( I_{2}y-\kappa y\right)
+2\left( I_{1}y\right) \left( I_{1}\dot{y}\right) \Big)+\lambda _{2}^{2}\dot{%
y}^{2}\left( I_{1}\dot{y}\right) \bigg)+O\left( a^{5}\right) ,
\end{eqnarray*}%
where%
\begin{equation*}
I_{3}z\left( t\right) :=\int_{t}^{1}\int_{\tau }^{1}z(\eta )\dot{z}%
^{2}\left( \tau \right) \mathrm{d}\eta \,\mathrm{d}\tau .
\end{equation*}%
In the second step, expression%
\begin{equation*}
\frac{1}{1-\kappa \lambda _{02}-\kappa \Delta \lambda _{2}}=\frac{1}{%
1-\kappa \lambda _{02}}+\kappa \Delta \lambda _{2}\frac{1}{\left( 1-\kappa
\lambda _{02}\right) ^{2}}+\left( \kappa \Delta \lambda _{2}\right) ^{2}%
\frac{1}{\left( 1-\kappa \lambda _{02}\right) ^{3}}+O\left( \left( \Delta
\lambda _{2}\right) ^{3}\right)
\end{equation*}%
is used to obtain operator $N$ as%
\begin{align}
N& (\lambda _{0}+\Delta \lambda ,ay_{L}+O(a^{3}))  \notag \\
& \!\!\!\!=\frac{a^{3}}{2}\Bigg(\frac{1}{1-\kappa \lambda _{02}}\bigg(%
\lambda _{01}\Big(I_{3}y_{L}+\dot{y}_{L}^{2}\left( I_{2}y_{L}-2\kappa
y_{L}\right) \Big)+\lambda _{02}\dot{y}_{L}^{2}\left( I_{1}\dot{y}%
_{L}\right) \bigg)  \notag \\
& +\frac{\kappa }{\left( 1-\kappa \lambda _{02}\right) ^{2}}\bigg(2\lambda
_{01}^{2}\dot{y}_{L}\left( I_{1}y_{L}\right) \left( I_{2}y_{L}-\kappa
y_{L}\right) +\lambda _{01}\lambda _{02}\dot{y}_{L}\Big(\dot{y}_{L}\left(
I_{2}y_{L}-\kappa y_{L}\right) +2\left( I_{1}y_{L}\right) \left( I_{1}\dot{y}%
_{L}\right) \Big)+\lambda _{02}^{2}\dot{y}_{L}^{2}\left( I_{1}\dot{y}%
_{L}\right) \bigg)\Bigg)  \notag \\
& +O\left( \Delta \lambda _{1}a^{3},\Delta \lambda _{2}a^{3},a^{4},\left(
\Delta \lambda _{2}\right) ^{2}a^{3},\Delta \lambda _{1}\Delta \lambda
_{2}a^{3},a^{5}\right) ,  \label{en}
\end{align}%
and operator $\tilde{L}$ as 
\begin{align*}
\tilde{L}& (\lambda _{0}+\Delta \lambda )\left( ay_{L}+O(a^{3})\right) \\
& =a\bigg(-\frac{\Delta \lambda _{1}\left( I_{2}y_{L}-\kappa y_{L}\right)
+\Delta \lambda _{2}I_{1}\dot{y}_{L}}{1-\kappa \lambda _{02}}-\kappa \Delta
\lambda _{2}\frac{\lambda _{01}\left( I_{2}y_{L}-\kappa y_{L}\right)
+\lambda _{02}I_{1}\dot{y}_{L}}{\left( 1-\kappa \lambda _{02}\right) ^{2}} \\
& \qquad +\left( \kappa \Delta \lambda _{2}\right) ^{2}\frac{\lambda
_{01}\left( I_{2}y_{L}-\kappa y_{L}\right) +\lambda _{02}I_{1}\dot{y}_{L}}{%
\left( 1-\kappa \lambda _{02}\right) ^{3}}\bigg)+O\left( \Delta \lambda
_{1}a^{3},\Delta \lambda _{2}a^{3},\left( \Delta \lambda _{2}\right)
^{3}a,\Delta \lambda _{1}\left( \Delta \lambda _{2}\right) ^{2}a\right) .
\end{align*}

Such obtained $N$ and $\tilde{L}$ are inserted into (\ref{BifJed2}), so the
bifurcation equation reads%
\begin{equation}
\phi \left( a,\Delta \lambda \right) =c_{3}a^{3}+a\left( c_{11}\Delta
\lambda _{1}+c_{12}\Delta \lambda _{2}+c_{13}\left( \Delta \lambda
_{2}\right) ^{2}\right) +O\left( \Delta \lambda _{1}a^{3},\Delta \lambda
_{2}a^{3},\left( \Delta \lambda _{2}\right) ^{3}a,\Delta \lambda _{1}\left(
\Delta \lambda _{2}\right) ^{2}a,a^{4}\right) =0,  \label{BifJedF}
\end{equation}%
where constants are calculated as 
\begin{eqnarray}
c_{11} &=&-\frac{1}{1-\kappa \lambda _{02}}\left\langle I_{2}y_{L}-\kappa
y_{L},q_{l}^{\left( 2\right) }\right\rangle ,  \label{ce11} \\
c_{12} &=&-\frac{1}{1-\kappa \lambda _{02}}\left\langle I_{1}\dot{y}%
_{L},q\right\rangle -\frac{\kappa }{\left( 1-\kappa \lambda _{02}\right) ^{2}%
}\left\langle \lambda _{01}\left( I_{2}y_{L}-\kappa y_{L}\right) +\lambda
_{02}I_{1}\dot{y}_{L},q_{l}^{\left( 2\right) }\right\rangle ,  \label{ce12}
\\
c_{13} &=&\frac{\kappa ^{2}}{\left( 1-\kappa \lambda _{02}\right) ^{3}}%
\left\langle \lambda _{01}\left( I_{2}y_{L}-\kappa y_{L}\right) +\lambda
_{02}I_{1}\dot{y}_{L},q_{l}^{\left( 2\right) }\right\rangle ,  \label{ce13}
\\
c_{3} &=&\left\langle \tilde{N}\left( \lambda _{0}+\Delta \lambda
,ay_{L}+w\right) ,q_{l}^{\left( 2\right) }\right\rangle ,  \label{ce3}
\end{eqnarray}%
with $y_{L}$ being the normalized solution (\ref{Linearno}), $q_{l}^{\left(
2\right) }$ being given by (\ref{q2-l}), and $\tilde{N}$ being given through 
$N=a^{3}\tilde{N}+O,$ see (\ref{en}).

Function $\phi ,$ appearing in the bifurcation equation (\ref{BifJedF}), is
considered as a function of $a$ and a bifurcation parameter $\Delta \lambda
_{1},$ since $\Delta \lambda _{2}=\eta ^{\prime }\left( \text{$\lambda _{01}$%
}\right) \Delta \lambda _{1},$ see (\ref{delta lambda}). Proposition II.9.2
in \cite{GolubitskySchaeffer} requires that, calculated at $a=\Delta \lambda
_{1}=0,$ $\phi \left( a,\Delta \lambda \right) =\frac{\partial \phi \left(
a,\Delta \lambda \right) }{\partial a}=\frac{\partial ^{2}\phi \left(
a,\Delta \lambda \right) }{\partial a^{2}}=\frac{\partial \phi \left(
a,\Delta \lambda \right) }{\partial \Delta \lambda _{1}}=0$\ and $%
\varepsilon =\mathrm{sgn\,}\frac{\partial ^{3}\phi \left( a,\Delta \lambda
\right) }{\partial a^{3}},$ $\delta =\mathrm{sgn\,}\frac{\partial ^{2}\phi
\left( a,\Delta \lambda \right) }{\partial a\partial \Delta \lambda _{1}}.$
Straightforward calculation shows that $\phi ,$ given by (\ref{BifJedF}), is
strongly equivalent to%
\begin{equation*}
\varepsilon a^{3}+\delta \Delta \lambda _{1}a=0,\;\;\text{with}%
\;\;\varepsilon =\mathrm{sgn\,}c_{3},\;\delta =\mathrm{sgn\,}\left(
c_{11}+c_{12}\eta ^{\prime }\left( \text{$\lambda _{01}$}\right) \right) ,
\end{equation*}%
describing the pitchfork bifurcation, since by assumption $c_{3}\neq 0$ and $%
c_{11}+c_{12}\eta ^{\prime }\left( \text{$\lambda _{01}$}\right) \neq 0.$
\end{proof}

\begin{remark}
Constants defining the parameters $\varepsilon $ and $\delta $ in the case
of cantilevered rotating axially compressed local rod are reobtained for $%
\kappa =0,$ see Eq. (3.34) in \cite{Ata}.
\end{remark}

\section{The problem with imperfections \label{imperfect}}

The static stability problem, considered for the perfect rod in Section \ref%
{perfect}, is extended for the case of rod being initially deformed and
being loaded by the vertical force acting at its tip, i.e., for the rod with
imperfections in shape and loading. Angular velocity and intensity of the
horizontal force proved to be mutually dependent bifurcation parameters
causing multiple equilibrium configurations and it will be proved that
introduction of small initial deformation and small intensity of the
vertical force perturbs the pitchfork bifurcation obtained in the case of
perfect rod, i.e., these parameters correspond to the universal unfolding of
the perfect rod bifurcation problem.

Following the derivation procedure of equation (\ref{EqM}), given at the
beginning of Section \ref{perfect}, system of equations (\ref{6}), subject
to boundary conditions (\ref{7}), can be reduced to a single equation%
\begin{equation}
G\left( \lambda ,\alpha \right) y=0,\;\;\lambda =\left( \lambda _{1},\lambda
_{2}\right) \in 
%TCIMACRO{\U{211d} }%
%BeginExpansion
\mathbb{R}
%EndExpansion
^{2},\;\alpha =\left( \alpha _{1},\alpha _{2}\right) \in 
%TCIMACRO{\U{211d} }%
%BeginExpansion
\mathbb{R}
%EndExpansion
^{2},  \label{Gjednako0}
\end{equation}%
represented by the action of a non-linear operator on deflection given by%
\begin{equation}
G\left( \lambda ,\alpha \right) y:=\ddot{y}-\sqrt{1-\dot{y}^{2}}\frac{\alpha
_{1}\frac{1}{\rho _{0}}+\alpha _{2}J_{1}y+\lambda _{1}\left( J_{2}y-\kappa y%
\sqrt{1-\dot{y}^{2}}\right) +\lambda _{2}I_{1}\dot{y}}{1+\kappa \alpha _{2}%
\dot{y}+\kappa \lambda _{1}\dot{y}\left( I_{1}y\right) -\kappa \lambda _{2}%
\sqrt{1-\dot{y}^{2}}},  \label{OperatorG}
\end{equation}%
where, for $z\in L^{1}\left[ 0,1\right] ,$ 
\begin{equation*}
J_{1}z\left( t\right) :=\int_{t}^{1}\sqrt{1-\dot{z}^{2}\left( \tau \right) }%
\mathrm{d}\tau ,\;\;J_{2}z\left( t\right) :=\int_{t}^{1}\int_{\tau
}^{1}z(\eta )\sqrt{1-\dot{z}^{2}\left( \tau \right) }\mathrm{d}\eta \,%
\mathrm{d}\tau ,\;\;I_{1}z\left( t\right) :=\int_{t}^{1}z\left( \tau \right) 
\mathrm{d}\tau .
\end{equation*}%
The equation (\ref{Gjednako0}) is subject to boundary conditions 
\begin{equation}
y(0)=0,\;\;\dot{y}(0)=0,  \label{bc}
\end{equation}%
i.e., to (\ref{7})$_{2}$ and (\ref{7})$_{3}$, since other boundary
conditions are already used in obtaining (\ref{OperatorG}). Setting $\alpha
_{1}=\alpha _{2}=0$ in (\ref{OperatorG}), one obtains (\ref{opM}), i.e.,%
\begin{equation*}
G\left( \lambda ,0\right) y=M^{2}\left( \lambda ,y\right) .
\end{equation*}

The recognition problem for universal unfolding, represented by the question
whether the problem (\ref{Gjednako0}), subject to (\ref{bc}), for the
imperfect rod leads to the two-parameter universal unfolding of the function 
$\phi ,$ given by (\ref{BifJedF}), corresponding to the problem (\ref%
{Gjednako0}), subject to (\ref{bc}), for the perfect rod, will be addressed
using Proposition III.4.4. in \cite{GolubitskySchaeffer} in the following
theorem.

\begin{theorem}
\label{Ljapunov-Smit 2}Let $\left( \lambda _{0},0\right) $ be bifurcation
point obtained in Theorem \ref{Crandal} and let assumptions of Theorem \ref%
{Ljapunov-Smit} be satisfied. In addition, let%
\begin{equation}
\det \left( 
\begin{array}{cc}
d_{01} & d_{21} \\ 
d_{02} & d_{22}%
\end{array}%
\right) \neq 0,  \label{deter}
\end{equation}%
where $d_{01},$ $d_{02},$ $d_{21},$ and $d_{22}$ are given by (\ref{de0}), (%
\ref{de22}), and (\ref{de23}). Then, the problem (\ref{Gjednako0}), (\ref{bc}%
) can be reduced to an equation $\Psi (a,\lambda ,\alpha )=0$, given by (\ref%
{psaj}), which is a two-parameter universal unfolding of $\phi ,$ given by (%
\ref{BifJedF}), in the sense of Definition 1.3 in \cite{GolubitskySchaeffer}.
\end{theorem}

\begin{proof}
In order to obtain the two-parameter unfolding of function $\phi ,$ given by
(\ref{BifJedF}), the procedure for obtaining the bifurcation equation (\ref%
{BifJedF}), given in proof of Theorem \ref{Ljapunov-Smit}, is followed. The
analogue of bifurcation equation (\ref{BifJed}) reads%
\begin{equation}
Q\left( G\left( \lambda ,\alpha \right) y\right) =0.  \label{BifJed-G}
\end{equation}

First, the operator $G,$ given by (\ref{OperatorG}), is rewritten as%
\begin{equation*}
G\left( \lambda ,\alpha \right) y=L^{2}\left( \lambda \right) y+K\left(
\lambda ,\alpha ,y\right) ,
\end{equation*}%
with 
\begin{align*}
K\left( \lambda ,\alpha ,y\right) & :=G\left( \lambda ,\alpha \right)
y-L^{2}\left( \lambda \right) y \\
& =-\sqrt{1-\dot{y}^{2}}\frac{\alpha _{1}\frac{1}{\rho _{0}}+\alpha
_{2}J_{1}y+\lambda _{1}\left( J_{2}y-\kappa y\sqrt{1-\dot{y}^{2}}\right)
+\lambda _{2}I_{1}\dot{y}}{1+\kappa \alpha _{2}\dot{y}+\kappa \lambda _{1}%
\dot{y}\left( I_{1}y\right) -\kappa \lambda _{2}\sqrt{1-\dot{y}^{2}}}+\frac{%
\lambda _{1}\left( I_{2}y-\kappa y\right) +\lambda _{2}I_{1}\dot{y}}{%
1-\kappa \lambda _{2}}.
\end{align*}

Second, using Taylor's expansion of the operator $K$ in a neighborhood of $%
(\lambda _{0},0)$, i.e., for $\lambda =\lambda _{0}+\Delta \lambda $, $%
\left\vert \Delta \lambda \right\vert \ll 1$ and $y=ay_{L}+w(\lambda
,ay_{L}) $, with $w=O(a^{3}),$ the following expression is obtained%
\begin{align*}
K& (\lambda ,\alpha ,y) \\
& =\frac{1}{2\left( 1-\kappa \lambda _{2}\right) }\bigg(\alpha _{1}\frac{1}{%
\rho _{0}}\left( \dot{y}^{2}-2\right) +\alpha _{2}\Big(\left( 1-t\right)
\left( \dot{y}^{2}-2\right) +I_{1}\dot{y}^{2}\Big)+\lambda _{1}\Big(I_{3}y+%
\dot{y}^{2}\left( I_{2}y-2\kappa y\right) \Big)+\lambda _{2}\dot{y}%
^{2}\left( I_{1}\dot{y}\right) \bigg) \\
& \qquad +\frac{\kappa }{2\left( 1-\kappa \lambda _{2}\right) ^{2}}\bigg(%
\alpha _{1}\alpha _{2}\frac{1}{\rho _{0}}\dot{y}\left( 2+\dot{y}^{2}\right)
+\alpha _{2}^{2}\dot{y}\Big(\left( 1-t\right) \left( \dot{y}^{2}+2\right)
-I_{1}\dot{y}^{2}\Big) \\
& \qquad +2\alpha _{1}\lambda _{1}\frac{1}{\rho _{0}}\dot{y}\left(
I_{1}y\right) +\alpha _{1}\lambda _{2}\frac{1}{\rho _{0}}\dot{y}^{2}+2\alpha
_{2}\lambda _{1}\dot{y}\Big(\left( 1-t\right) \left( I_{1}y\right)
+I_{2}y-\kappa y\Big)+\alpha _{2}\lambda _{2}\dot{y}\Big(\left( 1-t\right) 
\dot{y}+2I_{1}\dot{y}\Big) \\
& \qquad +2\lambda _{1}^{2}\dot{y}\left( I_{1}y\right) \left( I_{2}y-\kappa
y\right) +\lambda _{1}\lambda _{2}\dot{y}\Big(\dot{y}\left( I_{2}y-\kappa
y\right) +2\left( I_{1}y\right) \left( I_{1}\dot{y}\right) \Big)+\lambda
_{2}^{2}\dot{y}^{2}\left( I_{1}\dot{y}\right) \bigg)+O\left( a^{4}\right) ,
\end{align*}%
implying 
\begin{align*}
K& \big(\lambda _{0}+\Delta \lambda ,\alpha ,ay_{L}+O(a^{3})\big) \\
& =-\frac{1}{1-\kappa \lambda _{02}}\left( \alpha _{1}\frac{1}{\rho _{0}}%
+\alpha _{2}\left( 1-t\right) \right) \\
& \qquad +a\frac{\kappa }{\left( 1-\kappa \lambda _{02}\right) ^{2}}\left(
\alpha _{1}\alpha _{2}\frac{1}{\rho _{0}}+\alpha _{2}^{2}\left( 1-t\right)
\right) \dot{y}_{L}+\Delta \lambda _{2}\frac{\kappa }{\left( 1-\kappa
\lambda _{02}\right) ^{2}}\left( \alpha _{1}\frac{1}{\rho _{0}}+\alpha
_{2}\left( 1-t\right) \right) \\
& \qquad +a^{2}\Bigg(\frac{1}{2\left( 1-\kappa \lambda _{02}\right) }\bigg(%
\alpha _{1}\frac{1}{\rho _{0}}\dot{y}_{L}^{2}+\alpha _{2}\Big(\left(
1-t\right) \dot{y}_{L}^{2}+I_{1}\dot{y}_{L}^{2}\Big)\bigg) \\
& \qquad \qquad +\frac{\kappa }{2\left( 1-\kappa \lambda _{02}\right) ^{2}}%
\bigg(2\alpha _{1}\lambda _{01}\frac{1}{\rho _{0}}\dot{y}_{L}\left(
I_{1}y_{L}\right) +\alpha _{1}\lambda _{02}\frac{1}{\rho _{0}}\dot{y}_{L}^{2}
\\
& \qquad \qquad +2\alpha _{2}\lambda _{01}\dot{y}_{L}\Big(\left( 1-t\right)
\left( I_{1}y_{L}\right) +I_{2}y_{L}-\kappa y_{L}\Big)+\alpha _{2}\lambda
_{02}\dot{y}_{L}\Big(\left( 1-t\right) \dot{y}_{L}+2I_{1}\dot{y}_{L}\Big)%
\bigg)\Bigg) \\
& \qquad +a\Delta \lambda _{2}\frac{\kappa }{\left( 1-\kappa \lambda
_{02}\right) ^{2}}\left( 1-\frac{\kappa }{1-\kappa \lambda _{02}}\right)
\left( \alpha _{1}\alpha _{2}\frac{1}{\rho _{0}}+\alpha _{2}^{2}\left(
1-t\right) \right) \dot{y}_{L} \\
& \qquad -\left( \Delta \lambda _{2}\right) ^{2}\frac{\kappa ^{2}}{\left(
1-\kappa \lambda _{02}\right) ^{3}}\left( \alpha _{1}\frac{1}{\rho _{0}}%
+\alpha _{2}\left( 1-t\right) \right) \\
& \qquad +a^{3}\Bigg(\frac{1}{2\left( 1-\kappa \lambda _{02}\right) }\bigg(%
\lambda _{01}\Big(I_{3}y_{L}+\dot{y}_{L}^{2}\left( I_{2}y_{L}-2\kappa
y_{L}\right) \Big)+\lambda _{02}\dot{y}_{L}^{2}\left( I_{1}\dot{y}%
_{L}\right) \bigg) \\
& \qquad \qquad +\frac{\kappa }{2\left( 1-\kappa \lambda _{02}\right) ^{2}}%
\bigg(\alpha _{1}\alpha _{2}\frac{1}{\rho _{0}}\dot{y}_{L}^{3}+\alpha
_{2}^{2}\dot{y}_{L}\Big(\left( 1-t\right) \dot{y}_{L}^{2}-\left( I_{1}\dot{y}%
_{L}^{2}\right) \Big) \\
& \qquad \qquad +2\lambda _{01}^{2}\dot{y}_{L}\left( I_{1}y_{L}\right) \Big(%
I_{2}y_{L}-\kappa y_{L}\Big)+\lambda _{01}\lambda _{02}\dot{y}_{L}\Big(\dot{y%
}_{L}\left( I_{2}y_{L}-\kappa y_{L}\right) +2\left( I_{1}y_{L}\right) \left(
I_{1}\dot{y}_{L}\right) \Big)+\lambda _{02}^{2}\dot{y}_{L}^{2}\left( I_{1}%
\dot{y}_{L}\right) \bigg)\Bigg) \\
& \qquad +a^{2}\Delta \lambda _{1}\frac{\kappa }{\left( 1-\kappa \lambda
_{02}\right) ^{2}}\bigg(\alpha _{1}\frac{1}{\rho _{0}}\dot{y}_{L}\left(
I_{1}y_{L}\right) +\alpha _{2}\dot{y}_{L}\Big(\left( 1-t\right) \left(
I_{1}y_{L}\right) +I_{2}y_{L}-\kappa y_{L}\Big)\bigg) \\
& \qquad +a^{2}\Delta \lambda _{2}\Bigg(\frac{\kappa }{2\left( 1-\kappa
\lambda _{02}\right) ^{2}}\alpha _{2}\bigg(2\dot{y}_{L}\left( I_{1}\dot{y}%
_{L}\right) -\left( I_{1}\dot{y}_{L}^{2}\right) \bigg) \\
& \qquad \qquad -\frac{\kappa ^{2}}{2\left( 1-\kappa \lambda _{02}\right)
^{3}}\bigg(\alpha _{1}\frac{1}{\rho _{0}}\dot{y}_{L}\Big(2\lambda
_{01}\left( I_{1}y_{L}\right) +\lambda _{02}\dot{y}_{L}\Big) \\
& \qquad \qquad +2\alpha _{2}\lambda _{01}\dot{y}_{L}\Big(\left( 1-t\right)
\left( I_{1}y_{L}\right) +I_{2}y_{L}-\kappa y_{L}\Big)+\alpha _{2}\lambda
_{02}\dot{y}_{L}\Big(\left( 1-t\right) \dot{y}_{L}+2I_{1}\dot{y}_{L}\Big)%
\bigg)\Bigg) \\
& \qquad -a\left( \Delta \lambda _{2}\right) ^{2}\frac{\kappa }{\left(
1-\kappa \lambda _{02}\right) ^{2}}\left( 1-\frac{2\kappa }{1-\kappa \lambda
_{02}}\right) \left( \alpha _{1}\alpha _{2}\frac{1}{\rho _{0}}+\alpha
_{2}^{2}\left( 1-t\right) \right) \dot{y}_{L} \\
& \qquad -\left( \Delta \lambda _{2}\right) ^{3}\frac{\kappa ^{3}}{\left(
1-\kappa \lambda _{02}\right) ^{4}}\left( \alpha _{1}\frac{1}{\rho _{0}}%
+\alpha _{2}\left( 1-t\right) \right) \\
& \qquad +O\left( a^{4},\Delta \lambda _{1}a^{3},\Delta \lambda
_{2}a^{3},\Delta \lambda _{2}^{2}a^{2},\Delta \lambda _{1}\Delta \lambda
_{2}a^{2},\left( \Delta \lambda _{2}\right) ^{3}a,\left( \Delta \lambda
_{2}\right) ^{4}\right) .
\end{align*}

Using the same arguments as for obtaining equation (\ref{BifJed2}), equation
(\ref{BifJed-G}) becomes 
\begin{equation*}
\left\langle \tilde{L}\left( \lambda \right) y+K\left( \lambda ,\alpha
,y\right) ,q\right\rangle =0,
\end{equation*}%
yielding the two-parameter unfolding of function $\phi $ in the following
form%
\begin{eqnarray}
&&\Psi \left( a,\Delta \lambda ,\alpha \right) 
\begin{tabular}{l}
=%
\end{tabular}%
\alpha _{1}d_{01}+\alpha _{2}d_{02}  \notag \\
&&\qquad +a\left( \alpha _{1}\alpha _{2}d_{11}+\alpha _{2}^{2}d_{12}\right)
+\Delta \lambda _{2}\left( \alpha _{1}d_{13}+\alpha _{2}d_{14}\right)  \notag
\\
&&\qquad +a^{2}\left( \alpha _{1}d_{21}+\alpha _{2}d_{22}\right) +a\Delta
\lambda _{1}c_{11}+a\Delta \lambda _{2}\left( c_{12}+\alpha _{1}\alpha
_{2}d_{23}+\alpha _{2}^{2}d_{24}\right) +\left( \Delta \lambda _{2}\right)
^{2}\left( \alpha _{1}d_{25}+\alpha _{2}d_{26}\right)  \notag \\
&&\qquad +a^{3}\left( c_{3}+\alpha _{1}\alpha _{2}d_{31}+\alpha
_{2}^{2}d_{32}\right) +a^{2}\Delta \lambda _{1}\left( \alpha
_{1}d_{33}+\alpha _{2}d_{34}\right) +a^{2}\Delta \lambda _{2}\left( \alpha
_{1}d_{35}+\alpha _{2}d_{36}\right)  \notag \\
&&\qquad +a\left( \Delta \lambda _{2}\right) ^{2}\left( c_{13}+\alpha
_{1}\alpha _{2}d_{34}+\alpha _{2}^{2}d_{35}\right) +\left( \Delta \lambda
_{2}\right) ^{3}\left( \alpha _{1}d_{51}+\alpha _{2}d_{52}\right)  \notag \\
&&\qquad +O\left( a^{4},\Delta \lambda _{1}a^{3},\Delta \lambda
_{2}a^{3},\Delta \lambda _{2}^{2}a^{2},\Delta \lambda _{1}\Delta \lambda
_{2}a^{2},\Delta \lambda _{1}\left( \Delta \lambda _{2}\right) ^{2}a,\left(
\Delta \lambda _{2}\right) ^{3}a,\left( \Delta \lambda _{2}\right)
^{4}\right) 
\begin{tabular}{l}
=%
\end{tabular}%
0,  \label{psaj}
\end{eqnarray}%
where constants are calculated as%
\begin{gather}
d_{01}=-\frac{1}{1-\kappa \lambda _{02}}\left\langle \frac{1}{\rho _{0}}%
,q\right\rangle ,\qquad d_{02}=-\frac{1}{1-\kappa \lambda _{02}}\left\langle
1-t,q\right\rangle ,  \label{de0} \\
d_{11}=\frac{\kappa }{\left( 1-\kappa \lambda _{02}\right) ^{2}}\left\langle 
\frac{1}{\rho _{0}}\dot{y}_{L},q\right\rangle ,\qquad d_{12}=\frac{\kappa }{%
\left( 1-\kappa \lambda _{02}\right) ^{2}}\left\langle \left( 1-t\right) 
\dot{y}_{L},q\right\rangle ,  \notag \\
d_{13}=-\frac{\kappa }{1-\kappa \lambda _{02}}d_{01},\qquad d_{14}=-\frac{%
\kappa }{1-\kappa \lambda _{02}}d_{02},  \notag
\end{gather}%
\begin{eqnarray}
d_{21} &=&\frac{1}{2\left( 1-\kappa \lambda _{02}\right) }\Bigg<\frac{1}{%
\rho _{0}}\dot{y}_{L}\bigg(\dot{y}_{L}+\frac{\kappa }{1-\kappa \lambda _{02}}%
\Big(2\lambda _{01}\left( I_{1}y_{L}\right) +\lambda _{02}\dot{y}_{L}\Big)%
\bigg),q\Bigg>,  \label{de22} \\
d_{22} &=&\frac{1}{2\left( 1-\kappa \lambda _{02}\right) }\Bigg<\left(
1-t\right) \dot{y}_{L}^{2}+I_{1}\dot{y}_{L}^{2}  \notag \\
&&+\frac{\kappa }{1-\kappa \lambda _{02}}\dot{y}_{L}\bigg(2\lambda _{01}\Big(%
\left( 1-t\right) \left( I_{1}y_{L}\right) +I_{2}y_{L}-\kappa y_{L}\Big)%
+\lambda _{02}\Big(\left( 1-t\right) \dot{y}_{L}+2I_{1}\dot{y}_{L}\Big)\bigg)%
,q\Bigg>,  \label{de23}
\end{eqnarray}%
\begin{gather*}
d_{23}=\left( 1-\frac{\kappa }{1-\kappa \lambda _{02}}\right) d_{11},\qquad
d_{24}=\left( 1-\frac{\kappa }{1-\kappa \lambda _{02}}\right) d_{12}, \\
d_{25}=\frac{\kappa ^{2}}{\left( 1-\kappa \lambda _{02}\right) ^{2}}%
d_{01},\qquad d_{26}=\frac{\kappa ^{2}}{\left( 1-\kappa \lambda _{02}\right)
^{2}}d_{02},
\end{gather*}%
\begin{gather*}
d_{31}=\frac{\kappa }{2\left( 1-\kappa \lambda _{02}\right) ^{2}}%
\left\langle \frac{1}{\rho _{0}}\dot{y}_{L}^{3},q\right\rangle ,\qquad
d_{32}=\frac{\kappa }{2\left( 1-\kappa \lambda _{02}\right) ^{2}}\bigg<\dot{y%
}_{L}\Big(\left( 1-t\right) \dot{y}_{L}^{2}-\left( I_{1}\dot{y}%
_{L}^{2}\right) \Big),q\bigg>, \\
d_{33}=\frac{\kappa }{\left( 1-\kappa \lambda _{02}\right) ^{2}}\left\langle 
\frac{1}{\rho _{0}}\dot{y}_{L}\left( I_{1}y_{L}\right) ,q\right\rangle
,\qquad d_{34}=\frac{\kappa }{\left( 1-\kappa \lambda _{02}\right) ^{2}}%
\bigg<\dot{y}_{L}\Big(\left( 1-t\right) \left( I_{1}y_{L}\right)
+I_{2}y_{L}-\kappa y_{L}\Big),q\bigg>,
\end{gather*}%
\begin{eqnarray*}
d_{35} &=&-\frac{\kappa ^{2}}{2\left( 1-\kappa \lambda _{02}\right) ^{3}}%
\bigg<\frac{1}{\rho _{0}}\dot{y}_{L}\Big(2\lambda _{01}\left(
I_{1}y_{L}\right) +\lambda _{02}\dot{y}_{L}\Big),q\bigg>, \\
d_{36} &=&\frac{\kappa }{2\left( 1-\kappa \lambda _{02}\right) ^{2}}\Bigg<2%
\dot{y}_{L}\left( I_{1}\dot{y}_{L}\right) -I_{1}\dot{y}_{L}^{2} \\
&&-\frac{\kappa }{1-\kappa \lambda _{02}}\dot{y}_{L}\bigg(2\lambda _{01}\Big(%
\left( 1-t\right) \left( I_{1}y_{L}\right) +I_{2}y_{L}-\kappa y_{L}\Big)%
+\lambda _{02}\Big(\left( 1-t\right) \dot{y}_{L}+2I_{1}\dot{y}_{L}\Big)\bigg)%
,q\Bigg>,
\end{eqnarray*}%
\begin{gather*}
d_{37}=-\left( 1-\frac{2\kappa }{1-\kappa \lambda _{02}}\right)
d_{11},\qquad d_{38}=-\left( 1-\frac{2\kappa }{1-\kappa \lambda _{02}}%
\right) d_{12}, \\
d_{39}=\frac{\kappa ^{3}}{\left( 1-\kappa \lambda _{02}\right) ^{3}}%
d_{01},\qquad d_{310}=\frac{\kappa ^{3}}{\left( 1-\kappa \lambda
_{02}\right) ^{3}}d_{02},
\end{gather*}%
while constants $c_{11},$ $c_{12},$ $c_{13},$ and $c_{3}$ are given by (\ref%
{ce11}), (\ref{ce12}), (\ref{ce13}), and (\ref{ce3}), respectively.

According to Proposition III.4.4 in \cite{GolubitskySchaeffer}, in order for 
$\Psi ,$ given by (\ref{psaj}), to be the two-parameter universal unfolding
of $\phi ,$ given by (\ref{BifJedF}), it is required that%
\begin{equation*}
\det \left( 
\begin{array}{cccc}
0 & 0 & \frac{\partial ^{2}\phi \left( a,\Delta \lambda \right) }{\partial
a\partial \Delta \lambda _{1}} & \frac{\partial ^{3}\phi \left( a,\Delta
\lambda \right) }{\partial a^{3}} \\ 
0 & \frac{\partial ^{2}\phi \left( a,\Delta \lambda \right) }{\partial
\Delta \lambda _{1}\partial a} & \frac{\partial ^{2}\phi \left( a,\Delta
\lambda \right) }{\partial \Delta \lambda _{1}^{2}} & \frac{\partial
^{3}\phi \left( a,\Delta \lambda \right) }{\partial \Delta \lambda
_{1}\partial a^{2}} \\ 
\frac{\partial \Psi \left( a,\Delta \lambda ,\alpha \right) }{\partial
\alpha _{1}} & \frac{\partial ^{2}\Psi \left( a,\Delta \lambda ,\alpha
\right) }{\partial \alpha _{1}\partial a} & \frac{\partial ^{2}\Psi \left(
a,\Delta \lambda ,\alpha \right) }{\partial \alpha _{1}\partial \Delta
\lambda _{1}} & \frac{\partial ^{3}\Psi \left( a,\Delta \lambda ,\alpha
\right) }{\partial \alpha _{1}\partial a^{2}} \\ 
\frac{\partial \Psi \left( a,\Delta \lambda ,\alpha \right) }{\partial
\alpha _{2}} & \frac{\partial ^{2}\Psi \left( a,\Delta \lambda ,\alpha
\right) }{\partial \alpha _{2}\partial a} & \frac{\partial ^{2}\Psi \left(
a,\Delta \lambda ,\alpha \right) }{\partial \alpha _{2}\partial \Delta
\lambda _{1}} & \frac{\partial ^{3}\Psi \left( a,\Delta \lambda ,\alpha
\right) }{\partial \alpha _{2}\partial a^{2}}%
\end{array}%
\right) \neq 0,
\end{equation*}%
where the partial derivatives are calculated at $a=\Delta \lambda _{1}=0.$
Straightforward calculation yields%
\begin{equation*}
\det \left( 
\begin{array}{cccc}
0 & 0 & c_{11}+\eta ^{\prime }\left( \text{$\lambda _{01}$}\right) c_{12} & 
6c_{2} \\ 
0 & c_{11}+\eta ^{\prime }\left( \text{$\lambda _{01}$}\right) c_{12} & 0 & 0
\\ 
d_{01} & 0 & \eta ^{\prime }\left( \text{$\lambda _{01}$}\right) d_{13} & 
2d_{21} \\ 
d_{02} & 0 & \eta ^{\prime }\left( \text{$\lambda _{01}$}\right) d_{14} & 
2d_{22}%
\end{array}%
\right) =-2\left( c_{11}+\eta ^{\prime }\left( \text{$\lambda _{01}$}\right)
c_{12}\right) ^{2}\det \left( 
\begin{array}{cc}
d_{01} & d_{21} \\ 
d_{02} & d_{22}%
\end{array}%
\right) \neq 0,
\end{equation*}%
due to assumption $\left( c_{11}+\eta ^{\prime }\left( \text{$\lambda _{01}$}%
\right) c_{12}\right) \neq 0$ of Theorem \ref{Ljapunov-Smit} and assumption $%
\det \left( 
\begin{array}{cc}
d_{01} & d_{21} \\ 
d_{02} & d_{22}%
\end{array}%
\right) \neq 0.$
\end{proof}

\begin{remark}
The condition for existence of universal unfolding in the case of
cantilevered rotating axially compressed local rod is reobtained from $\det
\left( 
\begin{array}{cc}
d_{01} & d_{21} \\ 
d_{02} & d_{22}%
\end{array}%
\right) \neq 0$ for $\kappa =0,$ see Eq. (4.9) in \cite{Ata}.
\end{remark}

\section{Numerical examples \label{calc}}

Theoretical results regarding the existence of bifurcation points,
occurrence of the pitchfork bifurcation for perfect rod and the
two-parameter unfolding corresponding to imperfect rod, given in Theorems %
\ref{Crandal}, \ref{Ljapunov-Smit}, and \ref{Ljapunov-Smit 2}, respectively,
are illustrated by the numerical examples. In particular, buckling mode
degeneration and post buckling shapes, along with the type of pitchfork
bifurcation, are numerically investigated.

The critical values $\lambda _{01}$ and $\lambda _{02},$ lying on the
interaction curve implicitly given by (\ref{frekventna}), along with the
trivial solution to equation (\ref{EqM}), subject to (\ref{bc2}), or
equation (\ref{EqM4}), subject to (\ref{bc4}), by Theorem \ref{Crandal}
represent the bifurcation points. The dependence of interaction curve shape
on non-locality parameter is reinvestigated in Figures \ref{krivaint1}, \ref%
{krivaint2}, and \ref{kis}. Namely, interaction curves for the first
buckling mode are monotonically decreasing functions for $\lambda _{01}\geq
0 $ and $\lambda _{02}\geq 0$ up to value of (dimensionless) non-locality
parameter $\kappa _{cr}=0.375325,$ as stated in \cite{AZ-1}. There is an
interaction curve branching at $\left( \lambda _{01},\lambda _{02}\right)
=\left( 29.145,0\right) $ for the critical value of the non-locality
parameter $\kappa _{cr},$ see Figure \ref{krivaint1}. If the non-locality
parameter has a value larger than $\kappa _{cr},$ then the interaction curve
branches for smaller value of $\lambda _{01},$ and higher value of $\lambda
_{02},$ as can be seen from Figure \ref{krivaint2}. 
\begin{figure}[tbph]
\begin{minipage}{0.65\columnwidth}
\centering
\includegraphics[width=0.68\columnwidth]{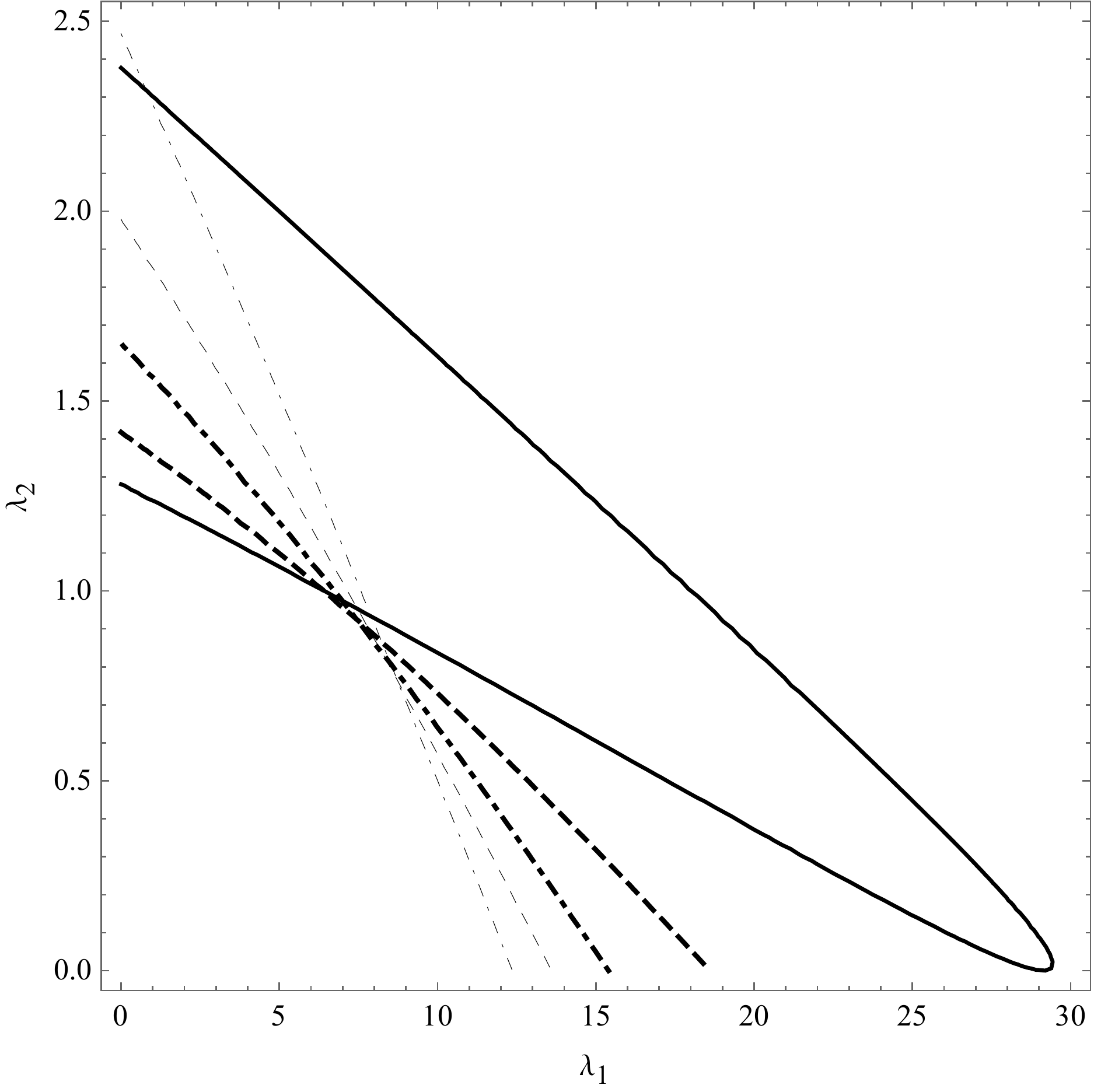}
\end{minipage}
\hfil
\begin{minipage}{0.35\columnwidth}
 \centering
 \begin{tabular}{l}
 $\protect\kappa=0$ - thin dot-dashed line\\
 $\protect\kappa=0.1$ - thin dashed line\\
 $\protect\kappa=0.2$ - thick dot-dashed line\\
 $\protect\kappa=0.3$ - thick dashed line\\
 $\protect\kappa_{cr}=0.375325$ - solid line
 \end{tabular}
\end{minipage}
\caption{Interaction curves for different values of non-locality parameter $%
\protect\kappa $.}
\label{krivaint1}
\end{figure}
\begin{figure}[tbph]
\begin{minipage}{0.65\columnwidth}
\centering
\includegraphics[width=0.68\columnwidth]{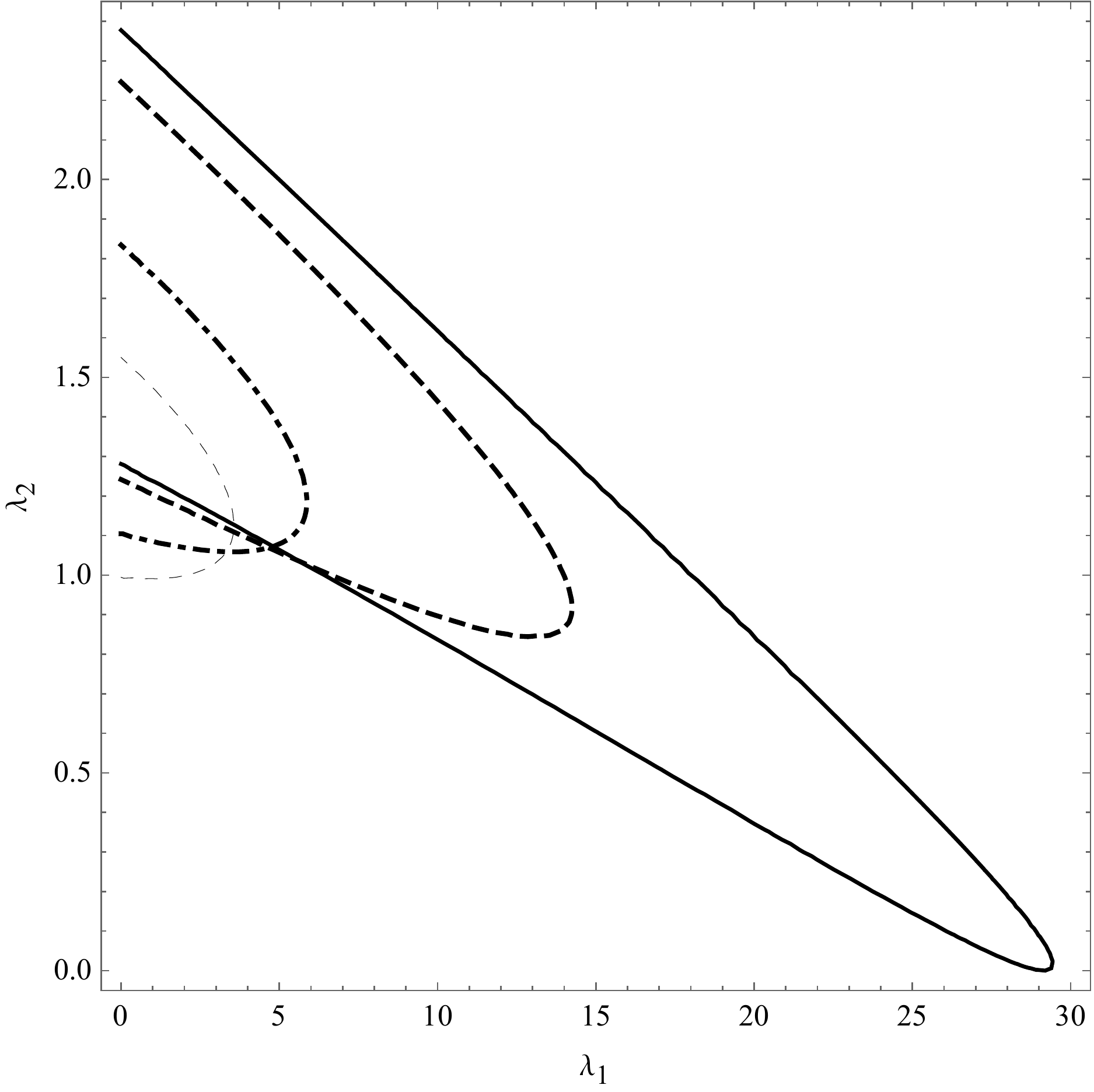}
\end{minipage}
\hfil
\begin{minipage}{0.35\columnwidth}
 \centering
 \begin{tabular}{l}
 $\protect\kappa=0.375325$ - solid line\\
 $\protect\kappa=0.4$ - thick dashed line\\
 $\protect\kappa=0.5$ - thick dot-dashed line\\
 $\protect\kappa=0.6$ - thin dashed line
 \end{tabular}
\end{minipage}
\caption{Interaction curves for different values of non-locality parameter $%
\protect\kappa $.}
\label{krivaint2}
\end{figure}

The occurrence of interaction curve branching is observed for higher modes
even if the non-locality parameter is less than the critical one, as in the
upper graph in Figure \ref{kis}, where $\kappa =0.25<\kappa _{cr}.$ If the
interaction curve branching occurs for the first mode, then it branches in
higher modes as well, see the lower graphs in Figure \ref{kis}. 
\begin{figure}[htbp]
\begin{center}
\begin{minipage}{0.45\columnwidth}
   \includegraphics[width=0.775\columnwidth]{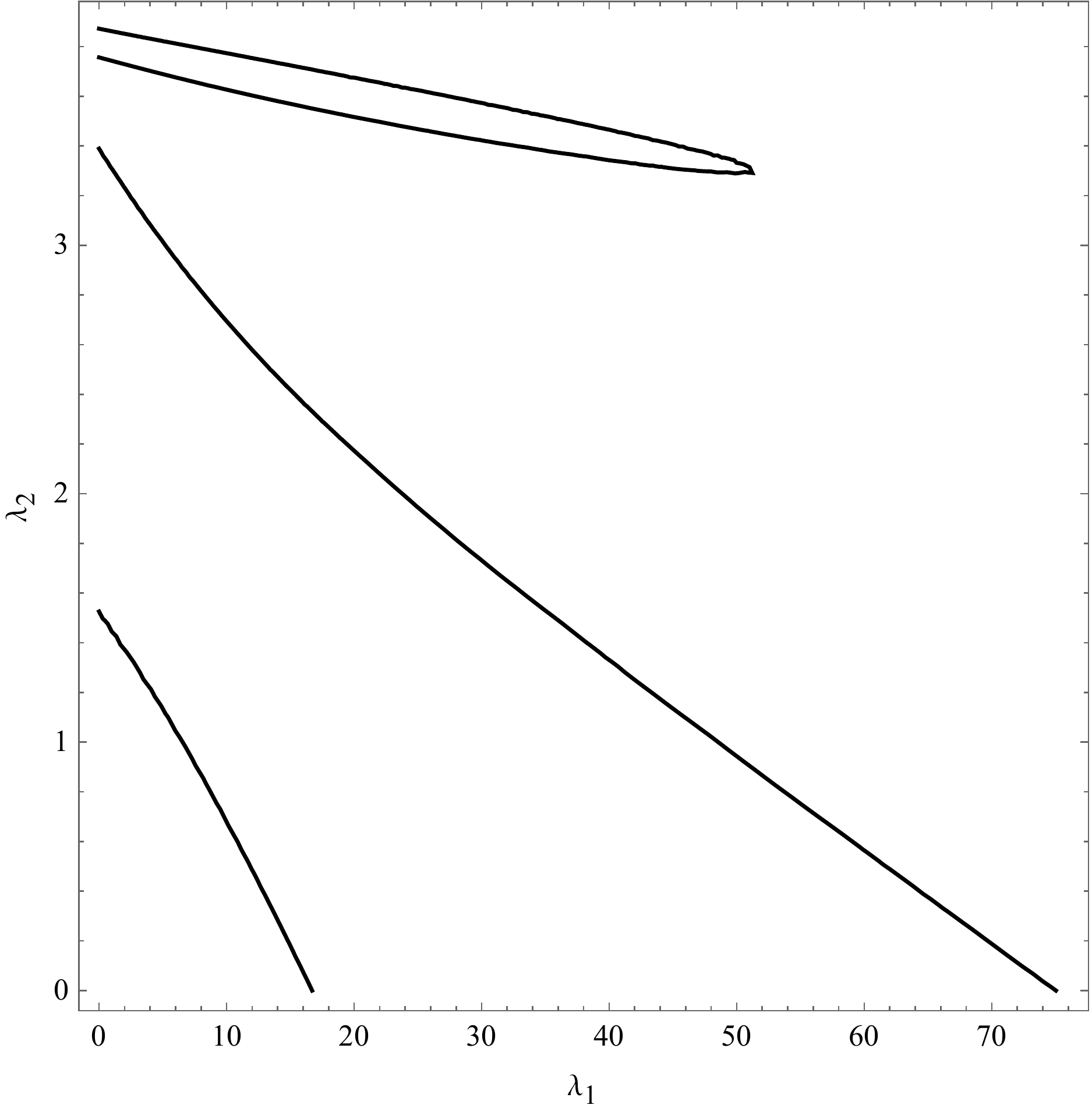}
  \end{minipage}
\vfil
\begin{minipage}{0.45\columnwidth}
   \includegraphics[width=0.775\columnwidth]{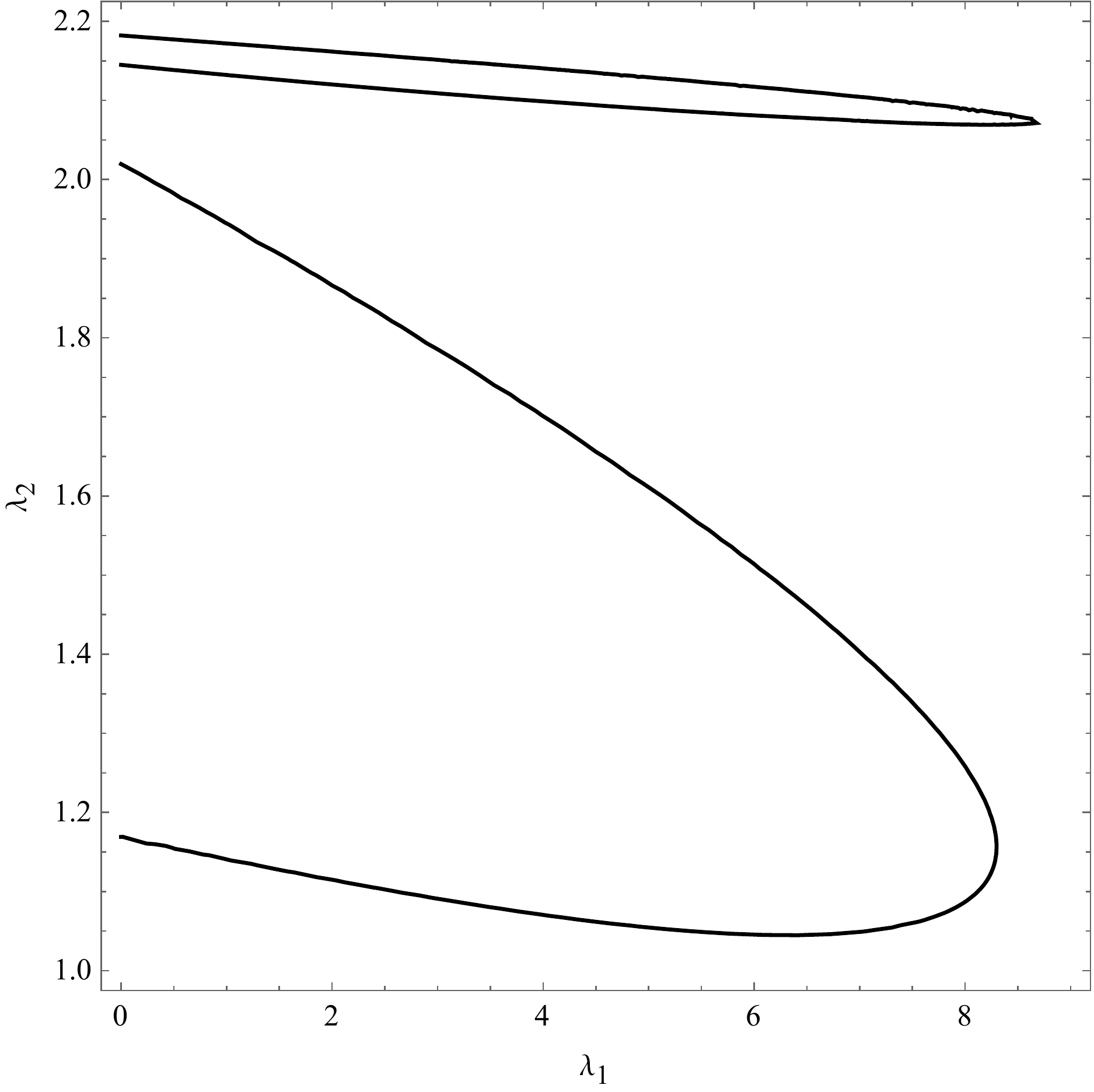}
  \end{minipage}
\hfil
\begin{minipage}{0.45\columnwidth}
   \includegraphics[width=0.775\columnwidth]{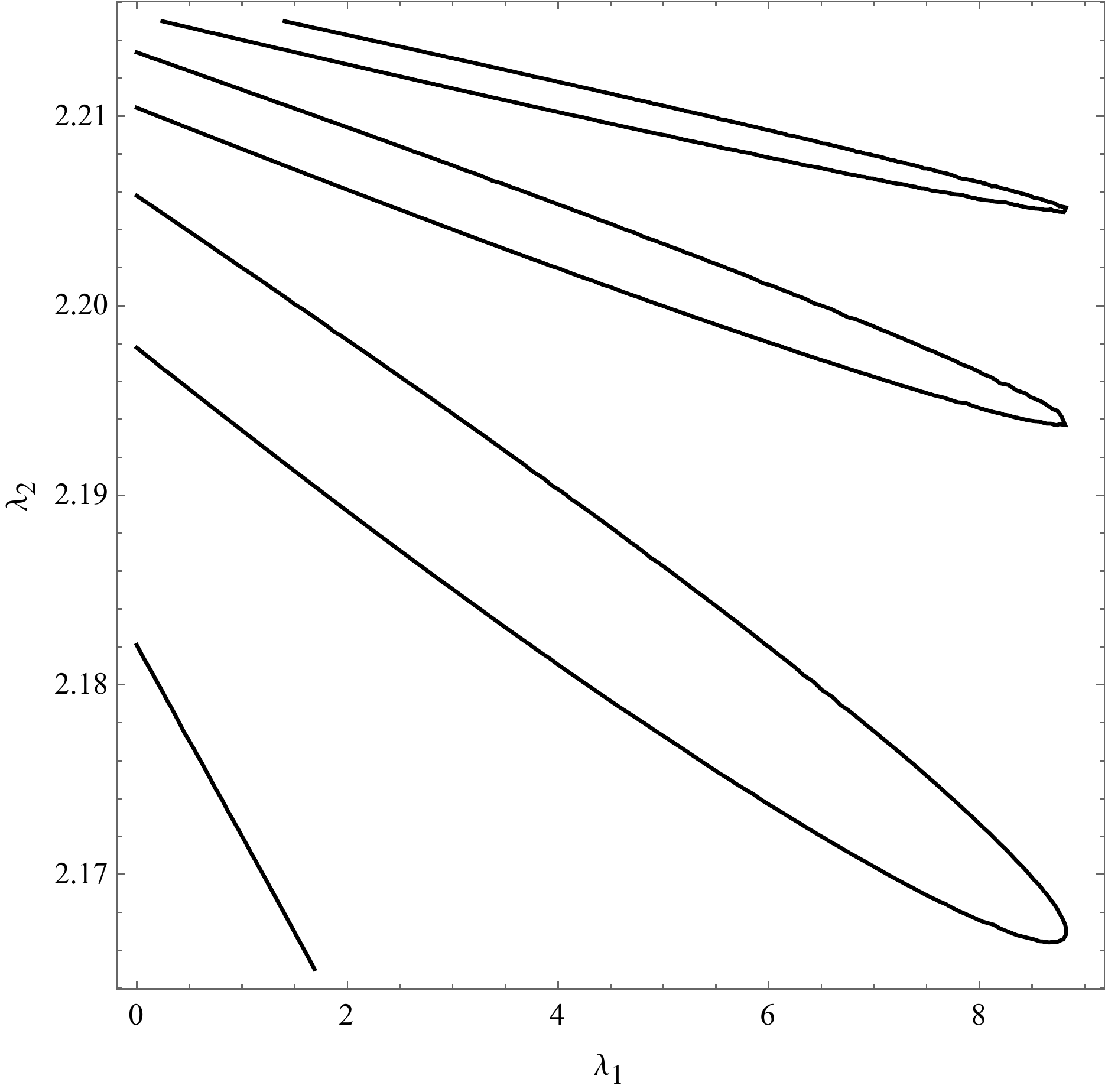}
  \end{minipage}
\end{center}
\caption{Interaction curves corresponding to different bucking modes for: $%
\protect\kappa =0.25$ - upper graphs; $\protect\kappa =0.45$ - lower left
graphs for lower-order modes; $\protect\kappa =0.45$ - lower right graphs
for higher-order modes.}
\label{kis}
\end{figure}

Figure \ref{linear} depicts the behavior of (to its maximum value)
normalized solution (\ref{Linearno}) of the linearized problem (\ref{EqL})
at the interaction curve branching point $\left( \lambda _{01},\lambda
_{02}\right) =\left( 8.29796,1.15665\right) $ for $\kappa =0.45,$ see also
the lower right graph in Figure \ref{kis}, as well as for the critical
values in its neighborhood on the lower (thick lines) and upper (thin lines)
branch. One notices that the shape of linear solution corresponding to the
first buckling mode is degenerating into the shape resembling to the second
buckling mode as the critical values pass from the lower to the upper branch
through the interaction curve branching point, as shown in Figure \ref%
{linear}. 
\begin{figure}[h!]
\begin{minipage}{0.65\columnwidth}
\centering
\includegraphics[width=0.8\columnwidth]{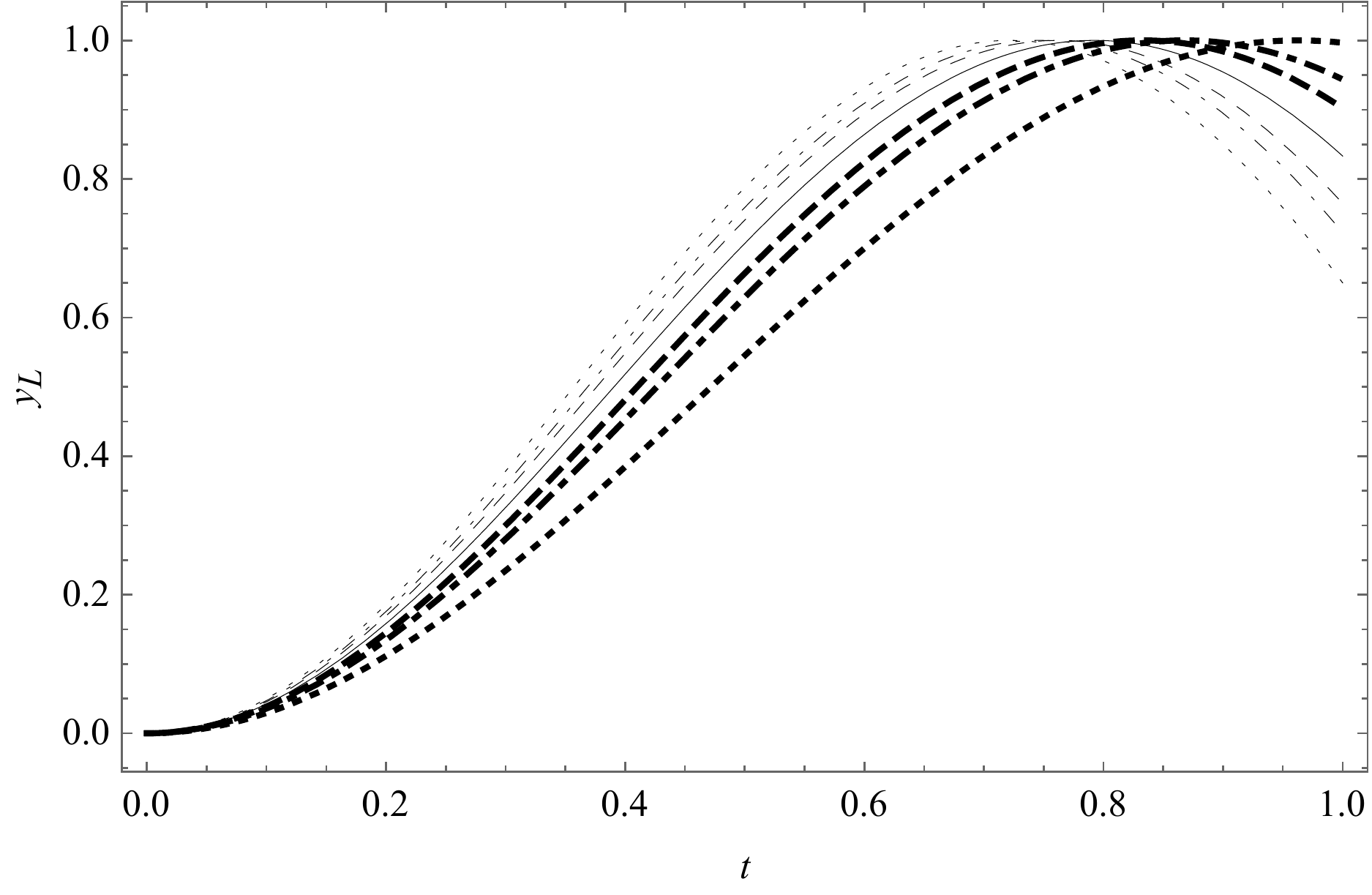}
\end{minipage}
\hfil
\begin{minipage}{0.35\columnwidth}
 \centering
 \begin{tabular}{l}
 $\!\!\!\! \left( \lambda _{01},\lambda _{02}\right)=$\\
 $(6,1.04544)$ - thick dotted line\\
 $(7.5,1.05978)$ - thick dot-dashed line\\
 $(8,1.08694)$ - thick dashed line\\
 $(8.29796,1.15665)$ - solid line\\
 $(8,1.25843)$ - thin dashed line\\
 $(7.5,1.33932)$ - thin dot-dashed line\\
 $(6,1.51419)$ - thin dotted line
 \end{tabular}
\end{minipage}
\caption{Plots of linear solution $y_{L}$ versus $t$ for non-locality
parameter $\protect\kappa =0.45$ for different critical values $(\protect%
\lambda _{01},\protect\lambda _{02})$.}
\label{linear}
\end{figure}

Post-critical buckling shapes, presented in Figures \ref{nelinearni-kapa-025}
- \ref{nelinearni-kapa-045-3}, are obtained as the numerical solution of
system of non-linear equations (\ref{6}), subject to boundary conditions (%
\ref{7}), with $\alpha _{1}=\alpha _{2}=0,\ $and either $\lambda
_{1}=\lambda _{01}+\Delta \lambda _{1}$ and $\lambda _{2}=\lambda _{02}+\eta
^{\prime }\left( \text{$\lambda _{01}$}\right) \Delta \lambda _{1},$ or $%
\lambda _{1}=\lambda _{01}+\bar{\eta}^{\prime }\left( \text{$\lambda _{02}$}%
\right) \Delta \lambda _{2}$ and $\lambda _{2}=\lambda _{02}+\Delta \lambda
_{2},$ where $\eta ^{\prime }$ is given by (\ref{eta prim}) and $\bar{\eta}%
^{\prime }$ is obtained analogously to $\eta ^{\prime }.$ In the each case
of post-buckling modes, the type of bifurcation point is determined
according to Theorem \ref{Ljapunov-Smit} by calculating $\varepsilon =%
\mathrm{sgn\,}c_{3}\;$and $\delta =\mathrm{sgn\,}\left( c_{11}+c_{12}\eta
^{\prime }\left( \text{$\lambda _{01}$}\right) \right) ,$ with $c_{11},$ $%
c_{12},$ and $c_{3}$ given by (\ref{ce11}), (\ref{ce12}), and (\ref{ce3}).
Using Theorem \ref{Ljapunov-Smit 2}, i.e., by calculating the determinant (%
\ref{deter}), it is also shown that in the each case of post-buckling modes
there exists the two-parameter unfolding for the initial displacement, i.e.,
curvature, assumed as 
\begin{equation*}
y\left( t\right) =t^{3}-\frac{4}{3}t^{2}+\frac{4}{9}t,\;\;\text{i.e.,}\;\;%
\frac{1}{\rho _{0}\left( t\right) }=\frac{6t-\frac{8}{3}}{\sqrt{1-\left(
3t^{2}-\frac{8}{3}t+\frac{4}{9}\right) ^{2}}},
\end{equation*}%
regardless of the use of $q_{l}^{\left( 2\right) },$ given by (\ref{q2-l}),
or $q_{l}^{\left( 4\right) },$ given by (\ref{q4-l}).

In the case of non-locality parameter $\kappa =0.25<\kappa _{cr},$ the first
post-buckling modes, corresponding to the different critical values on the
interaction curve from the upper graph in Figure \ref{kis}, are presented in
Figure \ref{nelinearni-kapa-025}. Their shape strongly resembles to the
shape of the first buckling mode of linear solution. Numerical calculation
of $\varepsilon =\mathrm{sgn\,}c_{3}$ and $\delta =\mathrm{sgn\,}\left(
c_{11}+c_{12}\eta ^{\prime }\left( \text{$\lambda _{01}$}\right) \right) $
shows that they are of different sign for both $q_{l}^{\left( 2\right) }$
and $q_{l}^{\left( 4\right) },$ implying the super-critical bifurcation. 
\begin{figure}[tbph]
\begin{minipage}{0.65\columnwidth}
\centering
\includegraphics[width=0.9\columnwidth]{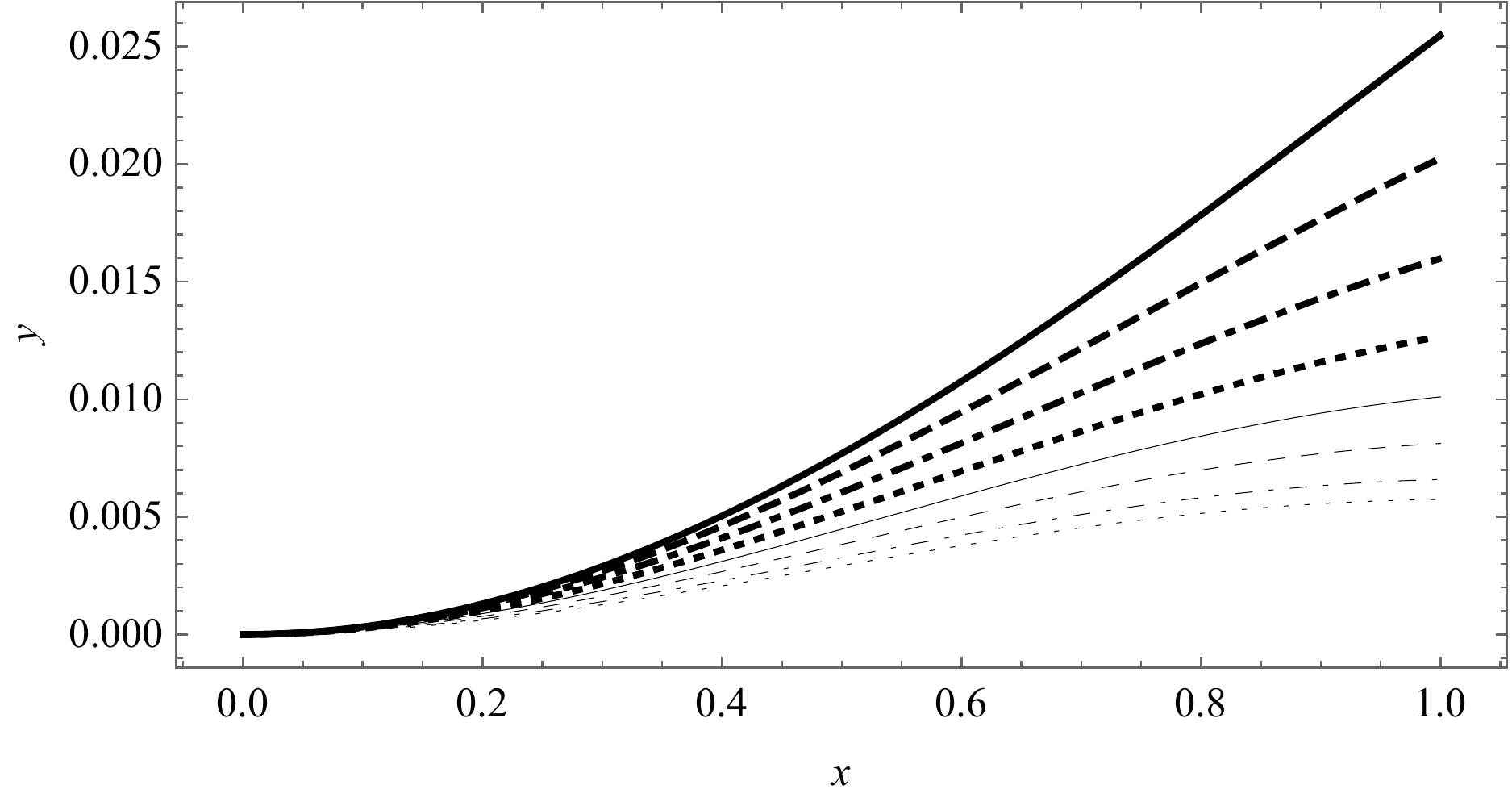}
\end{minipage}
\hfil
\begin{minipage}{0.35\columnwidth}
 \centering
 \begin{tabular}{l}
 $\!\!\!\! \left( \lambda _{01},\lambda _{02}\right)=$\\
 $(0.05,1.52248)$ - thick solid line\\
 $(2.5,1.33903)$ - thick dashed line\\
 $(5,1.13541)$ - thick dot-dashed line\\
 $(7.5,0.916144)$ - thick dotted line\\
 $(10,0.682732)$ - thin solid line\\
 $(12.5,0.436978)$ - thin dashed line\\
 $(15,0.180736)$ - thin dot-dashed line\\
 $(16.71310,0)$ - thin dotted line
 \end{tabular}
\end{minipage}
\caption{Plots of non-linear solution $y$ versus $x$ for non-locality
parameter $\protect\kappa =0.25$ for different critical values $(\protect%
\lambda _{01},\protect\lambda _{02})$, with $\Delta \protect\lambda _{1}=0.5$%
.}
\label{nelinearni-kapa-025}
\end{figure}

For the non-locality parameter value of $\kappa =0.45>\kappa _{cr}$ the
lower branch of interaction curve from the lower left graph in Figure \ref%
{kis} has a minimum at $\left( \lambda _{01}^{\left( \min \right) },\lambda
_{02}^{\left( \min \right) }\right) =\left( 6.32271,1.04474\right) $ at
which there is a distinct change in the shape of the first post-buckling
mode, as noticeable from Figure \ref{nelinearni-kapa-045-1}, since for $%
\lambda _{01}<\lambda _{01}^{\left( \min \right) }$ mode shapes resemble to
the shape of the first buckling mode of linear solution (thick lines), while
for $\lambda _{01}>\lambda _{01}^{\left( \min \right) }$ mode shapes
resemble more to the shape of the second buckling mode of linear solution
(thin lines). The post-buckling mode shapes at the minimum and in its
neighborhood are shown in Figure \ref{nelinearni-kapa-045-2}. Again, the
numerical calculation of $\varepsilon $ and $\delta $ shows that they are of
different sign for both $q_{l}^{\left( 2\right) }$ and $q_{l}^{\left(
4\right) },$ implying the super-critical bifurcation. 
\begin{figure}[tbph]
\begin{minipage}{0.65\columnwidth}
\centering
\includegraphics[width=0.9\columnwidth]{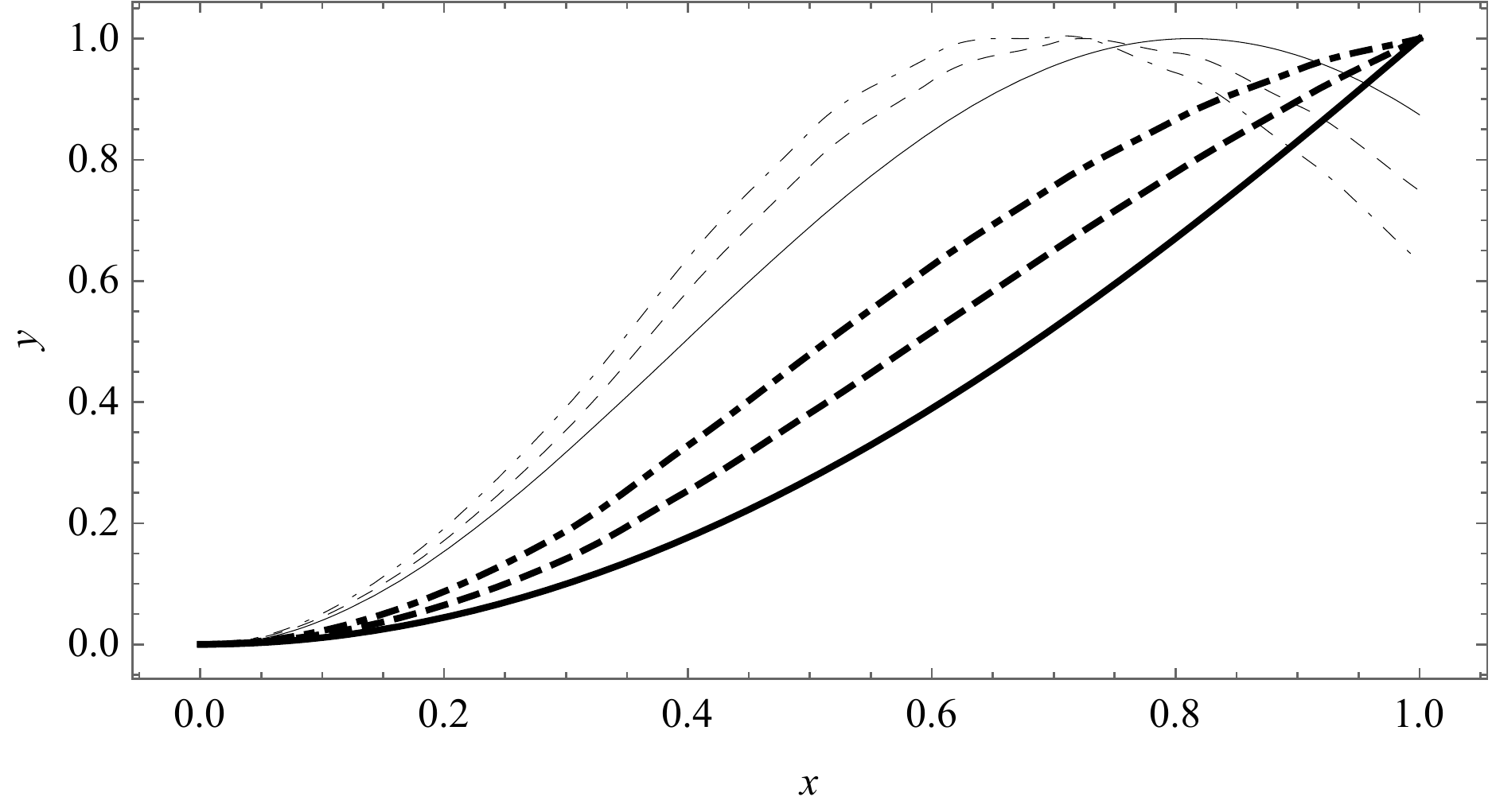}
\end{minipage}
\hfil
\begin{minipage}{0.35\columnwidth}
 \centering
 \begin{tabular}{l}
 $\!\!\!\! \left( \lambda _{01},\lambda _{02}\right)=$\\
 $(0.05,1.16776)$ - thick solid line\\
 $(2.5,1.10261)$ - thick dashed line\\
 $(5,1.05447)$ - thick dot-dashed line\\
 $(7,1.04881)$ - thin solid line\\
 $(7.5,1.05978)$ - thin dashed line\\
 $(8.29796,1.15665)$ - thin dot-dashed line
 \end{tabular}
\end{minipage}
\caption{Plots of non-linear solution $y$ versus $x$ for non-locality
parameter $\protect\kappa =0.45$ for different critical values $(\protect%
\lambda _{01},\protect\lambda _{02})$, with $\Delta \protect\lambda %
_{2}=0.02 $.}
\label{nelinearni-kapa-045-1}
\end{figure}
\begin{figure}[tbph]
\begin{minipage}{0.65\columnwidth}
\centering
\includegraphics[width=0.9\columnwidth]{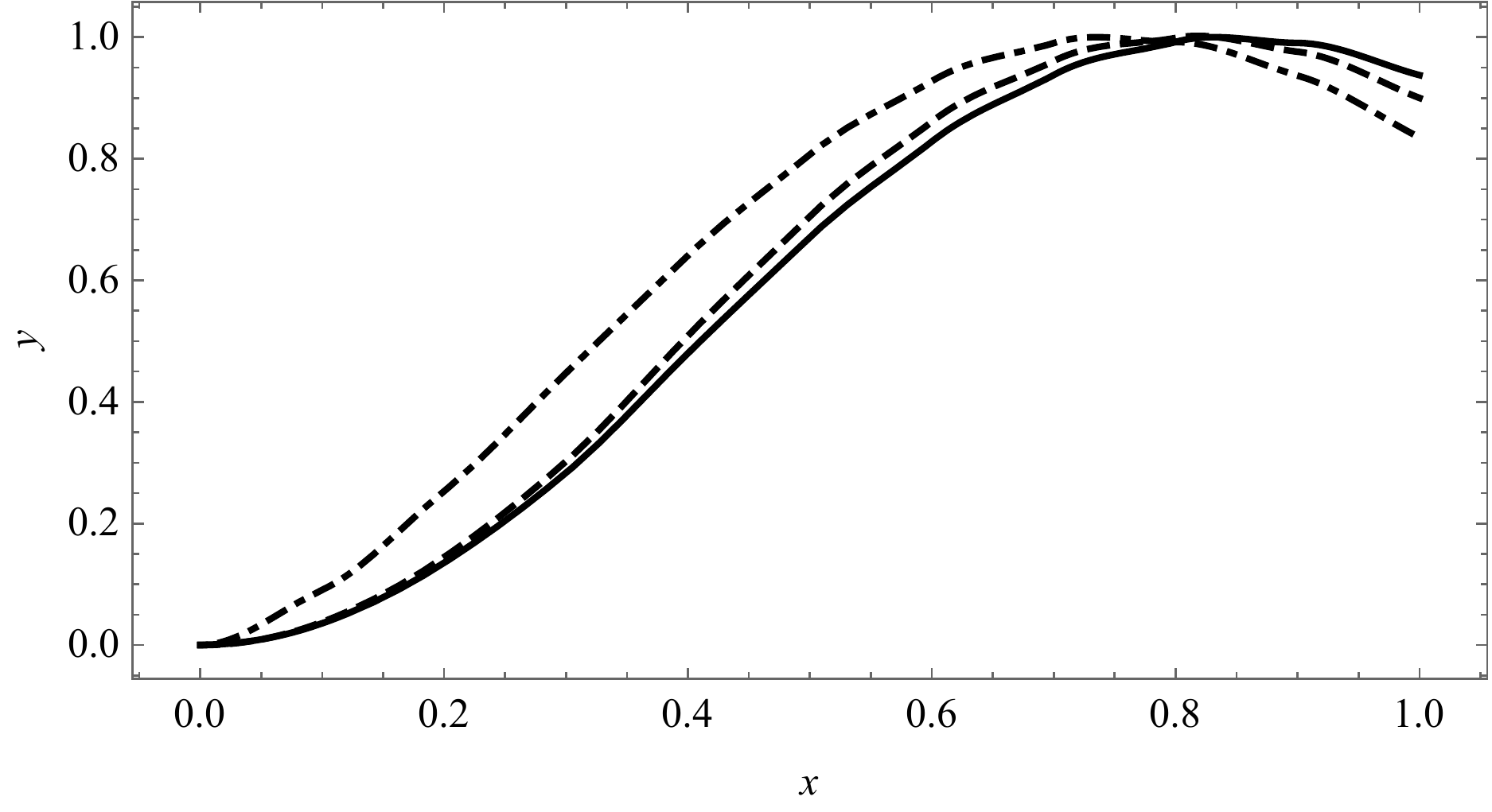}
\end{minipage}
\hfil
\begin{minipage}{0.35\columnwidth}
 \centering
 \begin{tabular}{l}
 $\!\!\!\! \left( \lambda _{01},\lambda _{02}\right)=$\\
 $(5.75,1.04684)$ - solid line\\
 $(6.32271,1.04474)$ - dashed line\\
 $(6.75,1.04684)$ - dot-dashed line
 \end{tabular}
\end{minipage}
\caption{Plots of non-linear solution $y$ versus $x$ for non-locality
parameter $\protect\kappa =0.45$ for different critical values $(\protect%
\lambda _{01},\protect\lambda _{02})$, with $\Delta \protect\lambda _{1}=0.5$%
.}
\label{nelinearni-kapa-045-2}
\end{figure}

The post-buckling mode shapes for the critical values lying on the upper
branch of interaction curve from the lower left graph in Figure \ref{kis}
are presented in Figure \ref{nelinearni-kapa-045-3}. Mode shapes
corresponding to $\left( \lambda _{01},\lambda _{02}\right) =\left(
7.5,1.33932\right) $ and $\left( \lambda _{01},\lambda _{02}\right) =\left(
5,1.61161\right) $ resemble to the shape of the first buckling mode of
linear solution and $\varepsilon $ and $\delta $ are of different sign
implying the super-critical bifurcation, while for $\left( \lambda
_{01},\lambda _{02}\right) =\left( 2.5,1.82714\right) $ and $\left( \lambda
_{01},\lambda _{02}\right) =\left( 0.05,2.01637\right) $ $\varepsilon $ and $%
\delta $ are of the same sign implying the sub-critical bifurcation. 
\begin{figure}[tbph]
\begin{minipage}{0.65\columnwidth}
\centering
\includegraphics[width=0.9\columnwidth]{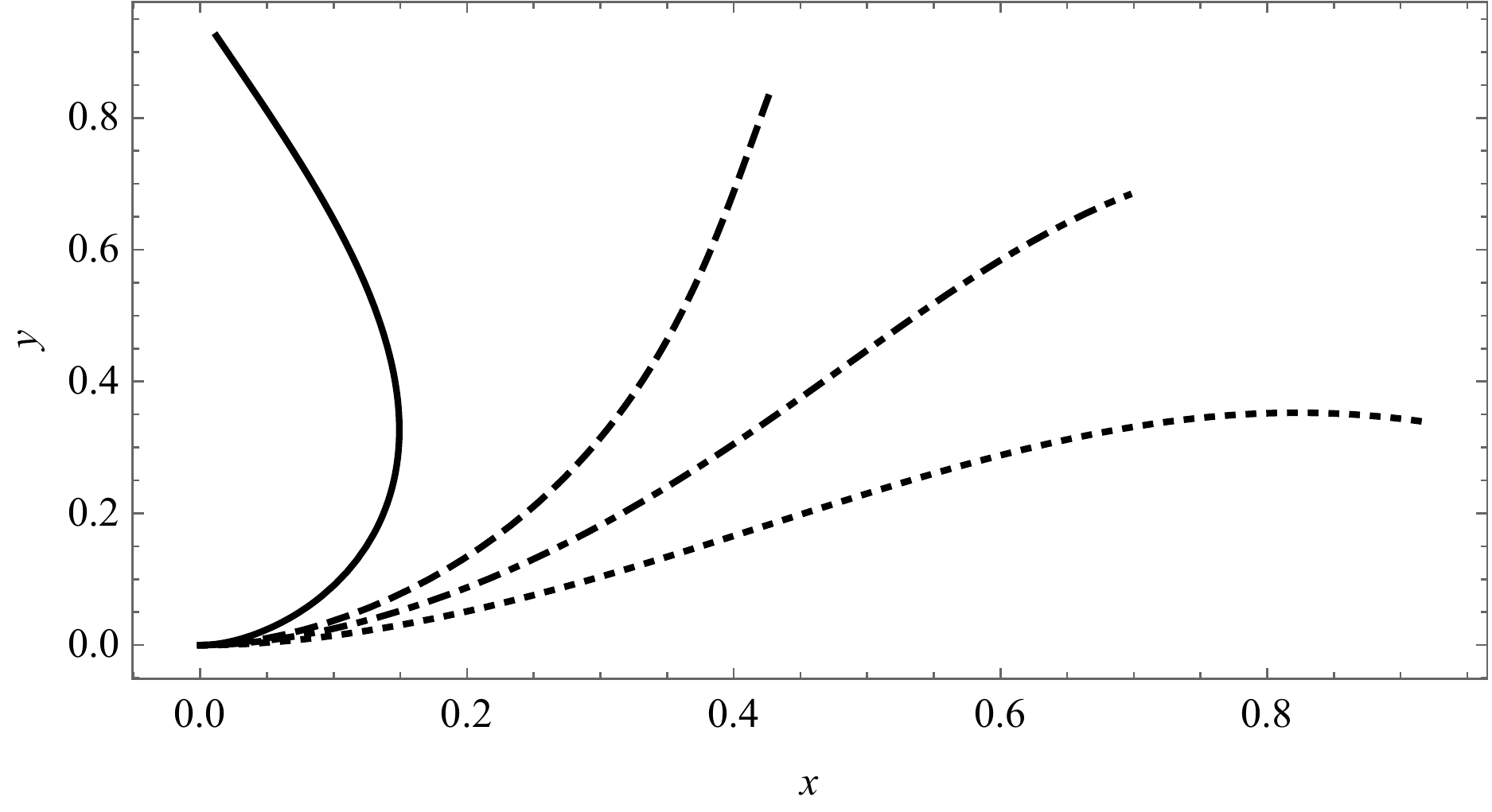}
\end{minipage}
\hfil
\begin{minipage}{0.35\columnwidth}
 \centering
 \begin{tabular}{l}
 $\!\!\!\! \left( \lambda _{01},\lambda _{02}\right)=$\\
 $(0.05,2.01637)$ - solid line\\
 $(2.5,1.82714)$ - dashed line\\
 $(5,1.61161)$ - dot-dashed line\\
 $(7.5,1.33932)$ - dotted line
 \end{tabular}
\end{minipage}
\caption{Plots of non-linear solution $y$ versus $x$ for non-locality
parameter $\protect\kappa =0.45$ for different critical values $(\protect%
\lambda _{01},\protect\lambda _{02})$, with $\Delta \protect\lambda %
_{2}=0.02 $.}
\label{nelinearni-kapa-045-3}
\end{figure}

\section{Conclusion}

Using the bifurcation theory, the static stability problem of cantilevered
rotating axially compressed non-local rod is revisited and extended by
considering imperfections in shape and loading, represented by the small
initial deformation of the rod and vertical force of small intensity acting
on rod's tip. The non-locality effects are included by considering the
stress gradient Eringen moment-curvature constitutive relation.

Theorem \ref{Crandal} states that the critical values of angular velocity
and intensity of the horizontal axial force acting on rod's tip, obtained
from the implicitly given interaction curve equation (\ref{frekventna}),
represent the bifurcation points for the non-linear equation (\ref{EqM4}),
subject to boundary conditions (\ref{bc4}). Theorem \ref{Ljapunov-Smit} uses
the Lyapunov-Schmidt reduction method and determines that the problem (\ref%
{EqM}), subject to (\ref{bc2}), admits pitchfork bifurcation, while Theorem %
\ref{Ljapunov-Smit 2} states that the selected imperfections constitute the
two-parameter universal unfolding of the same problem. The obtained results
in the case of Bernoulli-Euler moment-curvature constitutive equation reduce
to the results obtained in \cite{Ata}.

Numerical treatment of the interaction curve equation (\ref{frekventna})
shown the interaction curve branching in cases of small value of
non-locality parameter for higher modes even if the interaction curve is
monotone for the first (or second) mode. The interaction curve branching
occurs in cases of large value of non-locality parameter even for the first
mode (and higher modes as well). In the case of monotonically decreasing
interaction curve, using Theorem \ref{Ljapunov-Smit}, it is shown that the
pitchfork bifurcation is super-critical, which is also the case on lower
branch of interaction curve, while the bifurcation may change to
sub-critical on the upper branch. It is also shown, using Theorem \ref%
{Ljapunov-Smit 2}, that the selected initial deformation and vertical force
constitute the two-parameter universal unfolding.

\section*{Acknowledgement}

This work is supported by projects $174005$ and $174024$ of the Serbian
Ministry of Education, Science, and Technological Development and project $%
142-451-2384/2018$ of the Provincial Secretariat for Higher Education and
Scientific Research.

%\bibliographystyle{plain}
%\bibliography{dz}

\end{document}